\documentclass[10pt,conference]{IEEEtran}

\usepackage[normalem]{ulem}
\usepackage[utf8]{inputenc}
\usepackage{mathtools}
\usepackage{amsthm}
\usepackage{amssymb}
\usepackage{amsmath}
\usepackage{amsfonts}
\usepackage{url}
\usepackage{listings}
\usepackage{color}
\usepackage{xcolor}
\usepackage{graphicx}
\usepackage{caption}
\usepackage{breqn}
\usepackage{makecell}
\usepackage{enumitem}
\usepackage[ruled,linesnumbered]{algorithm2e}
\usepackage{balance}
\usepackage{multirow}
\usepackage{multicol}
\usepackage{bm}
\usepackage{subfigure}
\usepackage{soul}
\usepackage{etoolbox}
\usepackage{tikz}

\usepackage{tikz}

\newrobustcmd*\circled[1]{\tikz[baseline=(char.base)]{
            \node[shape=circle,draw,inner sep=1pt,fill,text=white,minimum size=1em] (char) {\textsf{\small #1}};}}

\newcommand{\parlabel}[1]{{\noindent\bf #1}}
\newcommand{\eg}{{\it e.g.}\xspace}
\newcommand{\ie}{{\it i.e.}\xspace}
\newcommand{\etc}{{\it etc.}\xspace}
\newcommand{\etal}{{\it et al.}\xspace}

\newtheorem{theorem}{Theorem}
\newtheorem{lemma}{Lemma}

\definecolor{darkpastelpurple}{rgb}{0.59, 0.44, 0.84}
\definecolor{darkgreen}{RGB}{0, 153, 67}

\newcommand{\del}[1]{}

\newcommand\name{\textit{ROTA-I/O}}
\newcommand\nameS{\textit{ROTA-Sched}}

\IEEEoverridecommandlockouts

\begin{document}




\title{\huge\name{}: Hardware/Algorithm Co-design for Real-Time I/O Control with Improved Timing Accuracy and Robustness\thanks{Corresponding author: Shuai Zhao, zhaosh56@mail.sysu.edu.cn.}}

\author{\IEEEauthorblockN{Zhe Jiang\IEEEauthorrefmark{2},
Shuai Zhao\IEEEauthorrefmark{3},
Ran Wei\IEEEauthorrefmark{5},
Xin Si\IEEEauthorrefmark{2},
Gang Chen\IEEEauthorrefmark{3},
Nan Guan\IEEEauthorrefmark{6} \\}
\IEEEauthorrefmark{2}SouthEast University, China,
\IEEEauthorrefmark{3}Sun Yat-Sen University, China,\\
\IEEEauthorrefmark{5}Lancaster University, United Kingdom,
\IEEEauthorrefmark{6}City University of Hong Kong, China \\}


\maketitle
    

\begin{abstract}
In safety-critical systems, timing accuracy is the key to achieving precise I/O control. 
To meet such strict timing requirements, dedicated hardware assistance has recently been investigated and developed.
However, these solutions are often fragile, due to unforeseen timing defects.
In this paper, we propose a robust and timing-accurate I/O co-processor, which manages I/O tasks using Execution Time Servers (ETSs) and a two-level scheduler.
The ETSs limit the impact of timing defects between tasks, and the scheduler prioritises ETSs based on their importance, offering a robust and configurable scheduling infrastructure.
Based on the hardware design, we present an ETS-based timing-accurate I/O schedule, with the ETS parameters configured to further enhance robustness against timing defects.
Experiments show the proposed I/O control method outperforms the state-of-the-art method in terms of timing accuracy and robustness without introducing significant overhead.
\end{abstract}

\section{Introduction}
\label{sc:Introduction}
Input/Outputs (I/Os) are a vital part of safety-critical systems~\cite{iso201126262,hennessy2011computer,jiang2018bluevisor}, as these systems rely on I/Os to interface with sensors and actuators that need to either perceive a hazard in time or make manoeuvres to avoid a hazardous situation~\cite{borgioli2022virtualization}.
For instance, spacecraft attitude control systems rely on precise and time-sensitive I/O operations to obtain and adjust the attitude of the spacecraft~\cite{chai2019six}, while autonomous vehicle engines depend on accurate I/O tasks for optimal fuel injection~\cite{mossinger2010software}.
A lack of I/O accuracy often leads to imprecise control of devices and less effective environmental readings being collected~\cite{zhao2020timing}.
Hence, it is necessary to ensure that I/O operations behave correctly: (i) with \emph{timing accuracy} -- being executed at exact time instants (or at least within a small margin) to achieve precise control~\cite{jiang2017gpiocp,zhao2020timing}; and (ii) with \textit{robustness} -- having the capability to deal with unexpected timing defects to maximise control accuracy~\cite{davis2007robust,burns2001real,davis2011survey}.

In industry, timing-accurate and robust I/O control is mandated by many globally recognised safety standards, including ISO 26262~\cite{iso201126262} and DO-178C~\cite{rierson2017developing}. 
As explicitly stated in clause 4 of ISO 26262, ``\textit{I/Os (e.g., sensors and actuators) must exhibit correct timing, accuracy, and robustness}''~\cite{iso201126262}. 
Additionally, clause 6 of the standard elaborates on the system-level failures that can arise from timing defects in I/O control, underscoring the importance of such requirements.

However, achieving timing accurate I/O control is ever-challenging due to unexpected timing defects, including ill-defined I/O-centric calculations~\cite{ballard2021machine}, underestimated WCET, and transient failures triggered by harsh environments, as recognised by Davids and Burns~\cite{davis2007robust} in typical safety-critical systems.
These unforeseen timing defects~\cite{burns2018robust} can fundamentally compromise I/O's timing accuracy by disrupting the well-arranged precise execution sequences of I/O tasks (\eg, produced by methods in~\cite{abdallah2016contention,kim2014integrated,betti2011real,kim2018supporting,jiang2017gpiocp, zhao2020timing}), and subsequently, cascading disruptions throughout the  I/O, leading to system failures and even catastrophic consequences~\cite{iso201126262}.

Software approaches that consider timing accuracy (\eg, \cite{kim2014integrated,betti2011real,kim2018supporting}) often provide carefully-planned I/O schedules on easily analysable platforms, \eg, uniprocessors and partitioned systems using isolation.
However, these methods have proven to be extremely difficult to achieve for both timing accuracy and robustness on modern platforms, due to ever-increasing hardware and architectural complexities. On such platforms, the significant uncertainty about I/O transmissions, \eg, communication delays and resource contentions~\cite{brandenburg2022multiprocessor,zhao2018thesis,audsley1993applying} occurring from software applications (instigation) to hardware devices (execution)~\cite{jiang2017gpiocp}, can directly disable any pre-planned schedule with an expected task arrival time.

Hardware approaches~\cite{TPUWeb,PRUWeb,jiang2017gpiocp,zhao2020timing} usually employ dedicated assistants to manage I/O operations in proximity to the devices.
By bringing stringent scheduling close to the devices, these approaches mitigate transmission uncertainty, ensuring timing accuracy under some use cases. 
However, to realise such management, these solutions rely on an idealistic assumption that ``\textit{the timing behaviours of the I/O operations must be well-defined}''~\cite{jiang2017gpiocp,zhao2020timing}. 
For systems with timing defects, the scheduling produced by the above approaches would be entirely violated (as explained above).
Thus, in more realistic application scenarios, it is important, but also challenging, to achieve timing-accurate and robust I/O management.

\begin{figure*}[t]
    \centering
    \includegraphics[width=1\textwidth]{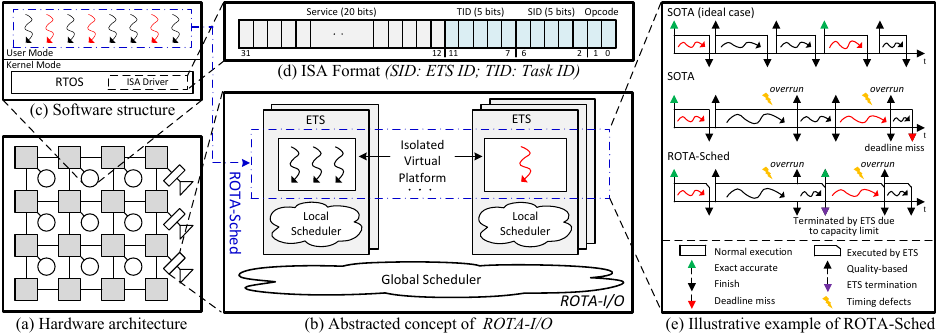}
    \caption{A conceptual overview of \name\ \emph{(square: router; circle: core; rectangle: \name; triangle: device; curved arrows: tasks with different timing guarantees; ETS: Execution Timer Server)}: (a) \name\ serves as a co-processor, enabling I/O management at hardware; (b) \name\ features sets of ETSs and a two-level scheduler, limiting the impact of timing defects while offering a configurable scheduling infrastructure; (c) \nameS{} schedules tasks and configures \name\ using (d) the dedicated ISA; (e) an illustrative mechanism of ROTA-Sched, improving I/O robustness and accuracy.}
    \label{fig:concept}
\end{figure*}


\parlabel{Contributions.}
In this paper, (i) we propose a device-coupled co-processor (\name{}) that manages I/O operations directly at the hardware level, featuring configurable Execution Time Servers (ETSs) and a scalable two-level scheduler.
The ETSs provide temporal isolation between I/O tasks, prohibiting defect propagation. The scheduler dispatches I/O operations hierarchically using the ETSs, establishing a scheduling infrastructure.
(ii) With the new hardware, we present an ETS-based scheduling method (namely \nameS{}) that schedules I/O tasks and configures the ETSs, achieving real-time I/O control with accuracy and robustness.
(iii) We built a systematic full-stack framework from the System-on-Chip (SoC) to the Instruction Set Architecture (ISA) and programming model, forming a complete solution for real-time I/O control.

We deployed our proposed solution on the Xilinx VC709 FPGA and examined it using various metrics.
Experimental results show that compared to State-of-the-Art (SOTA) real-time controllers, \name{} improves acceptance ratio and control quality by 22.61\% and 42.32\% respectively, on average (maximum: 2.18x and 3.06x). 
Synthesising a \name\ yields less than 3\% hardware overhead of a full-featured processor.

\parlabel{Organisation.} 
Sec.~\ref{sc:Overview} describes the overview of \name{}. Sec.~\ref{sc:Design} presents the architectural design of \name{}. Based on the hardware, the \nameS{} is proposed in Sec.~\ref{sc:Analysis}. Sec.~\ref{sc:Experiments} presents the experimental results. Finally, Sec.~\ref{sc:RelatedWork} provides the related work, and Sec.~\ref{sc:Conclusion} concludes the paper.

\section{\name: Overview}
\label{sc:Overview}

\subsection{Top-level Concepts}
\label{sbsc:Concepts}
To achieve timing-accurate and robust control for the I/Os, we adopted a hardware/algorithm co-designed approach (Fig.~\ref{fig:concept}) that tightly couples I/O management to the devices and enables hierarchical scheduling with temporal isolation.
Below, we detail the design and scheduling concepts.

\parlabel{Design concepts.} 
We designed the I/O co-processor (\name) with groups of Execution Time Servers (ETSs) and a two-level scheduler (Fig.~\ref{fig:concept}(b)).
The ETSs manage the I/O tasks and establish an isolated environment for their execution, effectively realising dedicated Virtual Platforms (VP) for task execution~\cite{lee2012realizing,shin2003periodic}.
With the two-level scheduler, I/O tasks are prioritised in a hierarchical manner: a \emph{global scheduler} distributes the time budget from the physical platform to the VPs and determines the \emph{scheduling parameters} (see scheduling concepts) of the VP in each ETS;
each ETS then has a \emph{local scheduler} to prioritise the tasks allocated to that VP using the distributed budget.
The design concepts have the following properties to ensure timing-accurate and robust I/O control: 

 
\begin{itemize}[align=parleft,labelindent=\parindent,leftmargin=*]
      \item \name~facilitates the I/O management being tightly coupled with the devices. 
      This effectively bypasses most of the uncertainty associated with I/O transmissions.
    
      \item The ETSs handle I/O tasks independently, which provides temporal isolation, hence the propagation of unexpected timing defects on the system is prevented.
    
      \item  The two-level scheduler is responsive to the system and environmental changes:
      for changes to one I/O task, the corresponding ETS can be configured and updated locally, without affecting the remaining systems;
      for changes to the taskset, the ETS-based allocation and schedule of I/O tasks can be updated globally.
      
\end{itemize}

\begin{table}[t]
\caption{Examples of the ISA for \name\ control \emph{(Priv: 1 and 0 indicate the kernel and user modes, respectively)}.}
\centering
\resizebox{.999\columnwidth}{!}{%
\begin{tabular}{l|c|l}
\hline
\textbf{Instruction} & \textbf{Priv} & \multicolumn{1}{c}{\textbf{Description}}      \\ \hline
\textbf{\texttt{c.set}}, rs1, rs2     & 1             & Set rs1 budget to ETS rs2.                 \\ 
\textbf{\texttt{c.enr}}, rs1, rs2     & 1             & Enroll ETS rs2 with a start time of rs1.                 \\ \hline
\textbf{\texttt{p.ld}}, rs1, rs2      & 0             & Pre-load the task addressed in rs1 to ETS rs2.    \\
\textbf{\texttt{i.ld}}, rs1, rs2      & 0             & Imm-load the task addressed in rs1 to ETS rs2.    \\
\textbf{\texttt{i.run}}, rs1       & 0             & Run the pre-loaded task addressed in rs1. \\ 
\hline
\end{tabular}}
\label{table:ISA}
\end{table}

\parlabel{Scheduling concepts.} 
With the new hardware design, an ETS-based scheduling algorithm (\nameS) was developed to configure the ETSs and coordinate the executions of I/O tasks on the ETSs.
For a given I/O taskset, the algorithm determines the parameters for the ETSs, including the start time,  budget, \etc{} 
This information is used by the global scheduler to manage the executions of the ETSs. 
The task execution order in each ETS is then produced by \name{}.
If the budget of an ETS has been exhausted due to timing defects, the scheduler will terminate its tasks to mitigate the impact of timing defects on the following ETSs, so that both robustness and timing accuracy can be guaranteed. 

\subsection{ISA Support}
\label{sbsc:ISASupport}
In coping with the reconfigurable characteristics of \name, we developed a dedicated ISA to abstract control interfaces for the software, where the instructions are classified into three categories (see Tab.~\ref{table:ISA}): c-type instructions (\texttt{c.x()}) for ETSs' configurations, 
i-type instructions (\texttt{i.x()}) for tasks requiring immediate-loading, and 
p-type instructions (\texttt{p.x()}) for tasks requiring pre-loading.
For instance, \texttt{c.set()} and \texttt{c.enr()} respectively assign the time budget and the start offset of an ETS, whereas both \texttt{p.ld()} and \texttt{i.ld()} are used for loading I/O tasks. 
The main distinction between these loading instructions lies in their operational scope:
\texttt{i.ld()} includes both the I/O operations and their scheduling parameters, facilitating immediate task execution. 
By contrast, \texttt{p.ld()} loads only the I/O operations, with task execution delayed until the scheduling parameters are defined by \texttt{
i.run()} (see Sec.~\ref{sbsc:ets_config} for details).
Given that the \texttt{c.set()} can cause device contention and bus congestion, it is designed as a privileged instruction, executable only by an operating system (OS) or a hypervisor with a global view of the system.

To ensure ISA coding and decoding efficiency, we standardised a uniform format for these instructions. 
Fig.\ref{fig:concept}(d) demonstrates the ISA format in four sections: the bottom 2 bits serve as the opcode, indicating the instruction type; the subsequent 10 bits represent the operated ETS and tasks; and finally, the top 20 bits convey specific service information.

\subsection{Programming Model}
\label{sc:ProgrammingModel}

With the new ISA introduced above, \name{} can be programmed through three phases, \ie, initialisation, execution, and adaptation.
In the initialisation phase, \nameS{} first determines the ETS parameters, \name{} then configures the ETSs using \texttt{c.set()} and \texttt{c.enr()}, with I/O tasks being pre-loaded via \texttt{p.ld()}.
In the execution phase, processors send run-time I/O tasks to \name\ using \texttt{i.ld()} or trigger the pre-loaded tasks using  \texttt{i.run()}, \name\ then manages these tasks according to the given configurations. 
Moreover, when a change occurs in the environment or the system, \eg, the increased arrival rate of a task, an adaptation will be triggered, where \nameS{} is invoked to identify the affected ETSs and reconfigure their scheduling parameters using \texttt{c.set()}.

\subsection{Integrating \name\ into a SoC}
\label{sbsc:SystemArchitecture}
Deploying \name\ shifts I/O management from the OS kernel (as in conventional embedded/computer architectures) to the hardware, resulting in architectural changes.
Fig.\ref{fig:concept}(a) depicts the integration of \name\ in a multi-/many-core SoC.
Specifically,  the \name\ is physically connected to the device, establishing a scheduling infrastructure for accurate and robust control.
We also connect \name~to the home port of a router via physical links to construct communication channels and timing synchronisation with processors.
Through these connections, tasks running on the processors can intercommunicate with \name\ utilising the ISA in Section \ref{sbsc:ISASupport}.

At the software level, we deploy a Real-Time Operating System (RTOS) in the kernel space, offering a real-time environment for applications that require timing guarantees (see Fig.~\ref{fig:concept}(c)). 
Given \name's comprehensive I/O management, we replace the default I/O manager within the RTOS using a new ISA-compatible driver, providing abstract access interfaces for the applications.
The implementation of the driver is straightforward. This forwards tasks and the scheduling parameters of ETSs (\eg, the period and the budget of an ETS) to \name~by interpreting the ISA.

\noindent Evidently, \name~provides the key to guaranteeing I/O timing accuracy and robustness against unexpected timing defects. 
We now detail the micro-architecture of \name.




\begin{figure}[t]
    \centering
    \includegraphics[width=1\columnwidth]{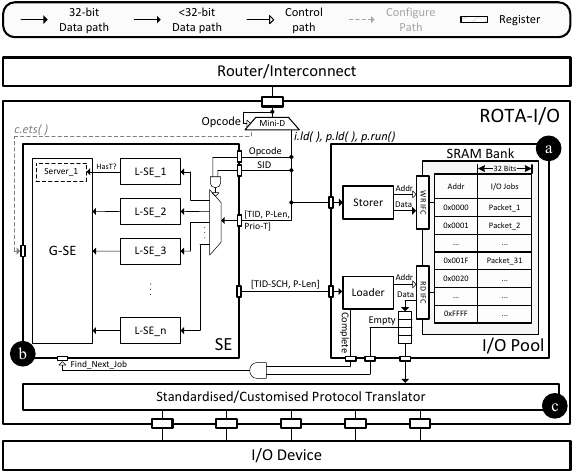}
    \caption{Top-level micro-architecture  of \name\ \emph{(Mini-D: mini decoder; Prio-S/T: priority of server/task; TID: task ID; TID-SCH: scheduled TID)}:
    \circled{a} tasks are maintained in an SRAM-based I/O pool, allowing prioritisation; \circled{b} nested schedulers prioritise the I/O tasks hierarchically; \circled{c} scheduled tasks are translated into physical signals for the device control.}
    \label{fig:top-micro}
\end{figure}

\section{\name: Micro-architecture}
\label{sc:Design}
The \name{} design introduces a set of ETSs and a two-level scheduler, limiting the propagation of timing defects and constructing a configurable scheduling infrastructure.
The top-level micro-architecture of \name{} is depicted in Fig.~\ref{fig:top-micro}, consisting of three key components: an I/O pool, a Scheduling Engine (SE), and a protocol translator.
The I/O pool (Fig.~\ref{fig:top-micro}.\circled{a}) buffers I/O tasks loaded by the processors and supports random access to them, allowing task prioritisation.
The SE (Fig.~\ref{fig:top-micro}.\circled{b}) creates two nested priority queues to schedule I/O tasks in a hierarchical manner.  
The translator (Fig.~\ref{fig:top-micro}.\circled{c}) interprets I/O tasks into specific control signals in the physical layer.
As the \name{} micro-architecture is designed to be compatible with various underlying protocol translators, either standardised or customised translators can be directly instantiated.
Next, we introduce the design details of the I/O pool and the SE.

\subsection{I/O Pool}
\label{sbsc:I/O Pool}
The micro-architecture of the I/O pool (Fig.~\ref{fig:top-micro}.\circled{a}) contains a dual-port SRAM, a pair of SRAM controllers and a FIFO queue.
The SRAM stores the I/O tasks, decomposed as specific I/O operations.
We use the memory address (12 bits) to index these I/O operations -- the upper 7 bits represent the Task ID (TID), and the lower 5 bits give the release order of  I/O operations with the same TID.
We connect the SRAM controllers to each  of the SRAM ports, behaving as a storer and a loader.
During execution,  when the I/O pool receives a \texttt{p.ld()} or an \texttt{i.run()} instruction, the payload will be written into the corresponding addresses of the SRAM using the storer.
In addition, when the I/O pool receives the task ID scheduled by the SE, the I/O operations will be read from the SRAM using the loader and are then pushed into the FIFO queue for subsequent execution by the protocol translator.

\begin{figure}[t]
    \centering
    \includegraphics[width=1\columnwidth]{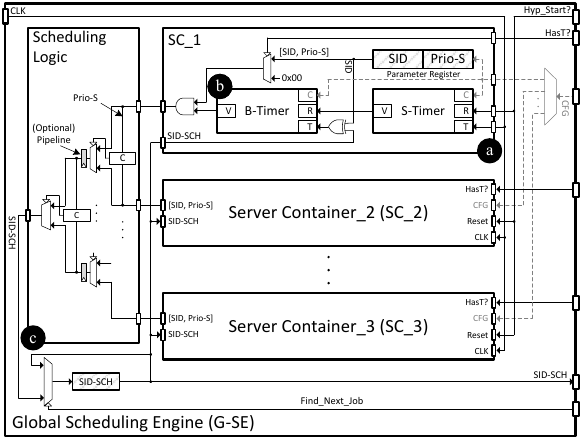}
    \caption{G-SE micro-architecture \emph{(refer to Fig.~\ref{fig:top-micro} for legends; SID-SCH: scheduled ETS ID; C: comparator)}: \circled{a} SCs are one-to-one associated to an ETS, featuring a parameter register to maintain the ETS's priority, and \circled{b} a pair of count-down timers to manage the budget; \circled{c} the ETSs's status is collected and prioritised, and the results are broadcasted to L-SEs.}
    \label{fig:G-SE}
\end{figure}

\subsection{Scheduling Engine (SE)}
\label{sbsc:SE}
The micro-architecture of the SE (Fig.~\ref{fig:top-micro}.\circled{b}) constructs two nested priority queues: an upper-level one, \ie, Global SE (G-SE), prioritises ETSs, deciding which ETS can execute the I/O task at a specific time point, while a lower-level one, \ie, Local SE (L-SE), prioritises the I/O tasks within an ETS, executing them according to their priority.

\parlabel{Global SE (G-SE).}
We designed the G-SE (Fig.~\ref{fig:G-SE}) using a collection of Server Containers (SCs) and scheduling logic.
The SCs are \emph{logically one-to-one} linked to with an ETS, maintaining the priority and time budget of the ETS.
In each SC (Fig.~\ref{fig:G-SE}.\circled{a}), a dedicated parameter register is employed to store the Server ID (SID) and priority of the ETS.
In addition,  a pair of count-down timers are deployed to manage the time budget assigned to the ETS ---  a Budget-Timer (B-Timer) holds the ETS's time budget, and a Start-Timer (S-Timer) maintains the refresh time point within each hyper-period (Fig.~\ref{fig:G-SE}.\circled{b}).
Inside each timer, two registers are deployed to store the reset and current values.
At the interfaces, three input ports and one output port are introduced. 
The input ports are used to program, reset and trigger the timer, and the output port reveals the status of the timer, \ie, whether the current value is greater than `0'. 
With that, the timer's reset value can be updated using the \texttt{c.cfg()} instruction through its program port; 
the timer's current value is reset when its reset port equals `0', and reduced by one when its trigger port meets a \emph{rising edge}. 
To refresh an ETS  with a preserved budget, the reset values in S-Timer and B-Timer are set as the refresh time point and the time budget (by \nameS\ in Sec.~\ref{sc:Analysis}).
We link the B-Timer's reset port to the S-Timer's status port, and the S-Timer's reset port to the global clock. 
This resets the S-Timer at the start of every hyper-period and resets the B-Timer when the S-Timer counts to `0'.

\begin{figure}[t]
    \centering
    \includegraphics[width=1\columnwidth]{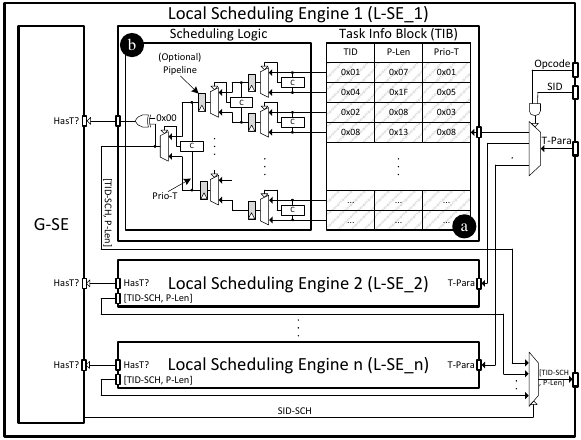}
    \caption{L-SE micro-architecture \emph{(refer to Fig.~\ref{fig:top-micro} for legends; T-Para: task parameters; TID-SCH: scheduled task ID)}: \circled{a} task parameters in the same ETS are stored in a TIB, implemented using a register chain; \circled{b} all entries of the TIB are compared, identifying the task with the highest priority and returning the TID to G-SE, forming the final scheduling decision.}
    \label{fig:L-SE}
\end{figure}

The scheduling logic collects the status of the ETSs and prioritises them.
To achieve this, the SC register is aggregated with the B-Timer output using an AND gate, which checks whether the SC is supplying enough time budget for its associated ETS.
If the ETS has enough time budget, the SID and priority of the ETS are returned; otherwise, 0 is returned.
The scheduling logic utilises pure combinational circuits to compare filtered results  (Fig.~\ref{fig:G-SE}.\circled{c}), selecting the SC and routing its SID to the L-SEs in a fixed cycle.

\parlabel{Local SE (L-SE).}
We designed the L-SE (Fig.~\ref{fig:L-SE}) using a Task Info Block (TIB) and scheduling logic.
The TIB records the parameters of the tasks buffered in the ETS, decomposed from the header of the \texttt{p.ld()} and \texttt{i.run()} instructions, including TID, P-Len and priority (Fig.~\ref{fig:L-SE}.\circled{a}).
We implement the TIB using a register chain to allow parallel accesses.
The scheduling logic deploys combinational circuits to collect the task parameters from the TIB and prioritises them (Fig.~\ref{fig:L-SE}.\circled{b}), returning the TID of the task with the highest priority.
Lastly, a multiplexer is employed to gather scheduling results from both the G-SE and L-SE. 
The outputs of the L-SEs and the G-SE are connected to the multiplexer's data ports and control port, respectively, forming the final scheduling decision.

\subsection{Design Trade-Offs}
\label{sbsc:DesignTrade-Offs}
Using a two-level micro-architecture to design the scheduler ensures  scalability, providing the opportunity to extend its capacity for handling more ETSs, so that the developer only needs to deploy additional L-SEs and SCs (in the G-SE) to manage the introduced contentions.
As the L-SE and SC are encapsulated in dedicated hardware modules and initialised independently, their integration will not introduce significant critical paths.
However, as increasingly more modules are integrated, critical paths may emerge due to the heightened complexity of prioritisation.
A practical solution to address this issue is to incorporate pipeline stages in the scheduling logic, as demonstrated in Fig.\ref{fig:G-SE}.\circled{c} and Fig.~\ref{fig:L-SE}.\circled{b}.

To further improve resource efficiency, it is possible to develop a \emph{monolithic} scheduler to manage all tasks collectively. 
However, this approach would result in a substantial increase in combinational logic, potentially causing critical paths.
This becomes particularly pronounced when scaling the scheduler to support more ETSs, ultimately compromising system-wide scalability.
In Sec.~\ref{sc:Scalability}, we provide a quantitative analysis demonstrating that the proposed micro-architecture can be scaled without leading to any critical paths.

In Sec.~\ref{sc:Scalability}, we provide a quantitative analysis demonstrating that the proposed micro-architecture can be scaled to support over 32 cores, without bringing any critical path.

\section{\nameS{}: An ETS-based Schedule}
\label{sc:Analysis}



This section presents the ETS-based schedule (\nameS{}) constructed based on the \name{}.
By utilising the ETSs provided by~\name{}, the \nameS{} allocates and schedules the I/O jobs in a hyperperiod on ETSs to mitigate the impact of interference caused by timing defects~\cite{shin2004compositional,davis2007robust,davis2009robust}, improving both timing accuracy and robustness of \name{}.
The \nameS{} is constructed based on the system model in Sec.~\ref{sbsc:model} with two major steps: (i)~allocation and scheduling of jobs in ETSs, and (ii) configuration of scheduling parameters of ETSs. 
The first step (Sec.~\ref{sbsc:alloc_sched}) produces a job-level schedule for I/O tasks which provides timing predictability and improved accuracy. 
The second step (Sec.~\ref{sbsc:ets_config}) determines the scheduling parameters for ETSs to further mitigate the impact of timing defects on I/O jobs, enhancing the robustness of \name{}.
With the two steps, we provide a complete ETS-based scheduling solution that achieves robust and timing-accurate I/O control. Notations introduced in this section are summarised in Tab.~\ref{tab:w_notations}.


\begin{figure}[!t]
\centering
\includegraphics[width=.85\columnwidth]{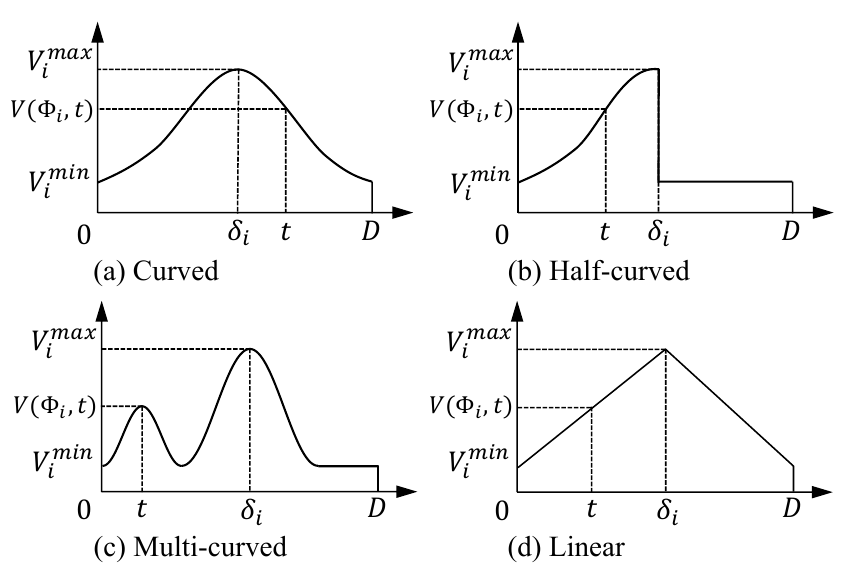}
\caption{Example timing-accurate models of I/O tasks \emph{(x-axis: time relative to the release of a task; y-axis: the resulting I/O control quality given a time instant $t$)}.}
\label{fig:vc}
\end{figure}

%

\subsection{System Model} \label{sbsc:model}
The system has $n$ periodic I/O tasks, denoted $\Gamma=\{\tau_1, \tau_2,...\tau_n\}$, where each task contains a series of sequential I/O operations on \name{}. 
An I/O task is defined by $\tau_i = (C_i, T_i, D_i, \delta_i, \Phi_i)$, in which $C_i$ is the worst-case execution time; $T_i$ is the period; $D_i$ is the deadline with $D_i = T_i$; $\delta_i$ denotes the ideal start offset of $\tau_i$ that achieves the maximum I/O operation quality, \ie, the exact timing accuracy; and $\Phi_i$ is the timing-accurate model associated with $\tau_i$.



In a hyper-period $T^H$ of the system, a task $\tau_i$ can raise up to $N_i = T^H \backslash T_i$ jobs, denoted as $\{\tau_i^1, \tau^2_i, ..., \tau^{N_i}_i\}$. Notation $\tau^j_i$ is the $j\textsuperscript{th}$ job of $\tau_i$, which has a release time of $r^j_i = T_i \times (j-1)$, a deadline of $d^j_i = r^j_i+D_i$, an ideal start offset of $\delta^j_i = T_i \times (j-1) + \delta_i$, and an actual offset $\theta_i^j$ determined by \nameS{}.
Jobs on \name{} are executed under a list schedule non-preemptively.

\parlabel{Timing-accurate Model.} Different from~\cite{zhao2020timing}, which uses a global timing-accurate model for all tasks to quantify the operation quality, tasks in this work can have different models. Fig.~\ref{fig:vc} presents four example timing-accurate models. For instance, Fig.~\ref{fig:vc}(a) shows a typical timing-accurate model~\cite{guerrra2008gravitational,guerra2009gravitational} that can be found in fuel injection in control systems~\cite{mossinger2010software} and the data sampling in automotive systems~\cite{liu2021real}. Fig.~\ref{fig:vc}(b) depicts a model for systems with more strict timing requirements, \eg, the radar scanning in a supersonic fighter~\cite{mei2017real}.
To reflect this generality, we let function $V(\Phi_i, t)$ denote the quality of $\tau_i$ (and its jobs) based on $\Phi_i$ and a start time $t$.
Task $\tau_i$ yields the maximum I/O control quality if it starts at the ideal time, \ie, $V^{max}_i = V(\Phi_i, \delta_i)$. 
We assume that $\Phi_i$ has a single time instant that leads to the exact timing accuracy, and all jobs of $\tau_i$ share the same $\Phi_i$.
In addition, a predictable I/O task can obtain a minimal quality of $V^{min}_i$ (\eg, $V^{min}_i = V(\Phi_i, D_i)$), and a deadline miss leads to zero quality or a negative penalty value.

\parlabel{ETS Model.} 
As described in Sec.~\ref{sc:Overview}, the tasks in the~\name{} are executed using the ETSs to prevent the propagation of timing defects between tasks in two ETSs. 
An ETS in \name{} is defined by $S_k =\{T_k,\lambda_k, \alpha_k\}$, in which $T_k$ gives the period, $\lambda_k$ denotes the budget, and $\alpha_k$ is the start time. 
Once $S_k$ is started, its budget is consumed with the passage of time. 
When the budget is exhausted, $S_k$ is terminated along with all its unfinished tasks. 
Each time when $S_k$ is released, the budget of $S_k$ is replenished to $\lambda_k$. 
To describe the relationship between ETSs and I/O jobs, we let function $G(S_k)$ denote a set of jobs that are allocated to $S_k$.

\begin{algorithm}[t]
$\vartriangleright$\text{ \texttt{\footnotesize Decomposing confliction graphs}}\\
\While{$||E|| > 0$}{
$\tau_i^j = \mathrm{argmax}_{i,j} \{\zeta^j_i ~|~\forall \tau_i^j \in \mathcal{G} \}$;\\
$\mathcal{G}_x = \mathcal{G} \setminus \tau_i^j $;\\
}
$\vartriangleright$\text{ \texttt{\footnotesize Schedule exact-accurate jobs on $S^*$}}\\
\For{{\normalfont each $\tau_i^j \in \mathcal{G}$, earliest $\delta_i^j$ first }}{
$S_k = \{\alpha_k = \delta_i^j,\lambda_k=C_i \}$; \\
$S^*=S^* \cup S_k$; \hspace{6pt} $G(S_k)=\{\tau_i^j\}$;\\
$\theta_i^j = \delta_i^j$; \hspace{33pt}
$\Gamma^H = \Gamma^H \setminus \tau_i^j $;
}

$\vartriangleright$\text{ \texttt{\footnotesize Schedule quality-based jobs on $S^\neg$}}\\
\normalfont generate each $S_k$ for $S^\neg$ using the free time spaces between the ETSs in $S^*$;\\

\For{{\normalfont each $S_k \in S^\neg$, the earliest $\alpha_k$ first}}{
    $\Gamma_k = \{ \tau_i^j ~|~ r_i^j < \alpha_k + \lambda_k \wedge d_i^j > \alpha_k , \forall \tau_i^j \in \Gamma^H\}$;\\
    \For{{\normalfont each $\tau_i^j \in \Gamma_k$, the earliest $d_i^j$ first}}{
        \If{\normalfont{$\tau_i^j$ is feasible with $G(S_k)$ on $S_k$}}{
                $G(S_k) = G(S_k) \cup \tau_i^j$;\\
                $\Gamma^H = \Gamma^H \setminus \tau_i^j $;\\
        }
    }
}

\normalfont \textbf{if} $\Gamma^H \neq \varnothing$, \textbf{return} $\mathrm{infeasible}$;\\

Optimise $\theta_i^j$ of all $\tau_i^j$ of each $S^k \in S^\neg$ based on $\Phi_i$ using a linear search;\\
\Return $\{S^*, S^\neg\}$;


\caption{ETS-based I/O allocation and schedule.}

\label{alg:sched}
\end{algorithm}

\subsection{Allocation and Scheduling of Jobs in ETSs}  \label{sbsc:alloc_sched}
For a given set of I/O jobs in a hyper-period (denoted $\Gamma^H$), \nameS{} produces a feasible ETS-based list schedule with guaranteed timing predictability and improved accuracy.
To achieve this, \nameS{} identifies the set of exact timing-accurate jobs that can yield the highest total quality (\ie, with exact timing accuracy), and assigns a dedicated ETS to guard the execution of each \textit{exact-accurate} job.
The remaining jobs are then scheduled to improve their quality based on $\Phi_i$, where several (\textit{quality-based}) jobs can share one ETS. 
This protects jobs in one ETS from the timing defects in other ETSs and limits the propagation of timing defects within this ETS.
Below, we first detail the identification of the timing-accurate jobs and then present the allocation and scheduling of jobs in the ETSs.
The constructed ETS-based allocation and scheduling process of I/O jobs is presented in Alg.~\ref{alg:sched}. Notations $S^*$ and $S^\neg$ denote the ETSs that manage the execution of the exact-accurate and quality-based jobs, respectively. 


\parlabel{Exact-accurate Jobs.}
To identify the exact-accurate jobs, a conflict graph $\mathcal{G}$ is formulated assuming each job $\tau_i^j \in \Gamma^H$ is executed at the ideal offset $\delta_i^j$.
Then, for two jobs with an execution conflict, an edge is added to indicate they cannot start at their $\delta_i^j$ simultaneously. Let $E(\tau_i^j)$ denote the jobs connected to $\tau_i^j$, the total quality of jobs that $\tau_i^j$ can affect if it starts at $\delta_i^j$ (represented by $\zeta_i^j$) is computed by Equation~\ref{eq:zeta}.
\begin{equation} \label{eq:zeta}
\zeta_i^j = \sum_{\tau_x^l \in E(\tau_i^j)} V^{max}_x
\end{equation}
With $\mathcal{G}$ and $\zeta_i^j, \forall \tau_i^j \in \Gamma^H$ obtained, the algorithm starts by decomposing $\mathcal{G}$ iteratively to identify the exact-accurate jobs based on $\zeta_i^j$ (lines 2-5). 
In each iteration, the $\tau_i^j$ with the maximum $\zeta_i^j$ is removed from $\mathcal{G}$, indicating it cannot be executed with exact timing accuracy, with the condition that $\tau_i^j$ can be fitted between the exact-accurate jobs (if possible). 
Otherwise, no feasible solution would be found during the allocation phase for the quality-based jobs.
This process repeats until jobs left in $\mathcal{G}$ are not connected with each other, \ie, $||E|| = 0$, which are identified as exact-accurate (\ie, $\theta_i^j=\delta_i^j$). This in general increases the number of exact timing-accurate jobs. For each exact-accurate job $\tau_i^j$, a dedicated ETS $S_k$ is assigned to manage its execution with a start time $\alpha_k = \delta_i^j$ and an initial budget $\lambda_k = C_k$ (lines 7-11). 

\begin{table}[t]
\caption{Notations introduced for constructing \nameS{}.}
\label{tab:w_notations}
\resizebox{\columnwidth}{!}{
\begin{tabular}{p{.15\columnwidth}p{.75\columnwidth}}
\hline
\textbf{Notation} & \textbf{Description} \\

\hline
$\Gamma$, $\Gamma^H$ & The sets of I/O tasks and their jobs in a hyperperiod.\\
$S^*$, $S^\neg$ & The sets of ETSs that manage the execution of exact-accurate and quality-based jobs.\\

\hline
$\tau_i$ & A periodic I/O task with an index of $i$.\\
$C_i$, $T_i$, $D_i$ & The worst-case execution time, period, and deadline $\tau_i$.\\
$\Phi_i$, $\delta_i$ & The timing-accurate model and the exact-accurate start offset of $\tau_i$.\\
$V(\Phi_i, t)$ & The quality given $\Phi_i$ and a start offset $t$ of $\tau_i$.\\
$V^{max/min}_i$ & The maximum/minimum quality of $\tau_i$.\\

\hline
$\tau_i^j$ & The $j$\textsuperscript{th} job of $\tau_i$ in a hyperperiod.\\
$r_i^j$, $d_i^j$, $\delta_i^j$ & The release time, deadline, and the exact-accurate (best-case) start offset of $\tau_i^j$.\\
$\zeta_i^j$ & Total quality from jobs affected by $\tau_i^j$ if it starts at $\delta_i^j$.\\
$\theta_i^j$ & The actual start offset of $\tau_i^j$ decided by \nameS{}.\\

\hline
$S_k$ & An ETS with an index $k$.\\
$\alpha_k$, $T_k$,  $\lambda_k$ & The start time, releasing period, and capacity of $S_k$.\\ 
\hline

$G(S_k)$ & The set of I/O jobs being assigned to $S_k$.\\
$||~\cdot~||$ & The size of a given set or list.\\
\hline
\end{tabular}}
\end{table}



\parlabel{Quality-based Jobs.}
With $S^*$ constructed, a set of ETSs for the quality-based jobs (\ie, $S^\neg$) are generated using the free time space between the ETSs in $S^*$. The $\alpha_k$ and $\lambda_k$ of each $S_k \in S^\neg$ are initialised based on the start time and length of the corresponding space (line 13).
For each $S_k \in S^\neg$, the algorithm takes the unallocated jobs that are active during $S_k$ (\ie, $\Gamma_K$ in line 15) and always tries to allocate the job with the earliest $d_i^j$ to $S_k$ (lines 16-21). A simple {feasibility test} is used to decide the acceptance or rejection of this allocation based on the remaining capacity of $S_k$, the allocated jobs $G(S_k)$, and $d_i^j$ (line 17). 
If feasible, $\tau_i^j$ is added to $S_k$ with the earliest possible start offset (lines 18-19). This in general increases the success ratio of the allocation. 
With the mapping of jobs decided, for each job in $S^\neg$ with the latest $\theta_i^j$ first, the algorithm postpones $\theta_i^j$ if a higher quality can be achieved based on $\Phi_i$ to further improve the operation quality of the system (where possible, without affecting other jobs using a linear search, line 24). The algorithm finishes and returns both the $S^*$ and $S^\neg$ with the allocation and scheduling of the I/O jobs embedded. Finally, the I/O jobs are registered to the associated ETS on \name{} using the instruction $\text{{\texttt{p.ld}}}, \tau_i^j, S_k$ introduced in Tab.~\ref{table:ISA}, in which $\tau_i^j \in G(S_k), \forall S_k \in S^* \cup S^\neg$.


\parlabel{Verification.}
The proposed \nameS{} is conducted on each I/O job in a hyper-period and returns a list schedule of the jobs, in which the schedule verifies whether a job can meet its deadline along with the scheduling process. Thus, similar to the job-level schedule in~\cite{chen2019timing,zhao2020timing}, the scheduling algorithm also serves as the timing verification method of the system, in which the system is schedulable if a feasible solution can be obtained by Alg.~\ref{alg:sched}, assuming there exist no timing defects. 


\parlabel{Discussion.}
The time complexity of Alg.~\ref{alg:sched} is $\mathcal{O}(n^2)$, as at most $(||\Gamma^H|| \times ||\Gamma^H||)$ iterations are required to deconstruct $\mathcal{G}$ (\ie, generation of $S^*$) and to allocate as well as schedule jobs in $S^\neg$.
The algorithm returns a feasible ETS-based I/O scheduling solution that improves timing accuracy. If an unexpected timing defect occurs, its impact on I/O jobs is effectively restricted to one ETS, preventing the cascading disruptions on other ETSs in the system. 
However, jobs within the ETS could miss their deadlines or be terminated (if the ETS capacity is exhausted), due to the timing defect. Below, we present the configuration of ETS parameters that further enhance the robustness by mitigating the impact of timing defects on the produced schedule.

\subsection{Configuration of ETS Parameters} \label{sbsc:ets_config}

\parlabel{Initial Configuration.} Based on the allocation and schedule of I/O jobs determined by Alg.~\ref{alg:sched}, the second step of \nameS{} produces the scheduling parameters of the ETSs, \ie, the start time $\alpha_k$, the capacity $\lambda_k$, and the period $T_k$ for each $S_k \in S^* \cup S^\neg$. The $\alpha_k$ is given by the earliest start offset of jobs in $G(S_k)$, \ie, $\alpha_k = \min\{ \theta_i^j ~|~ \forall \tau_i^j \in G(S_k) \}$.
An initial capacity of $S_k$ can be obtained by Equation~\ref{eq:lambda}, which is calculated based on (i) $\alpha_k$ and (ii) the latest finish time of jobs in $G(S_k)$.
As the ETSs are generated for a complete hyper-period, all ETSs in the system have the same period of $T_k = T^H$. In addition, as requested by the scheduler in~\name{}, an earlier ETS (or job) is assigned with a higher priority.
\begin{equation} \label{eq:lambda}
\lambda_k = \max\{  \theta_i^j + C_i ~|~ \forall \tau_i^j \in G(S_k)\} - \alpha_k
\end{equation}

\parlabel{Additional Capacity of $S_k$.} With the above parameters, the ETSs can be scheduled using a list scheduler (as described in Sec.~\ref{sbsc:SE}), in which $S_k$ is always dispatched at $\alpha_k$ and finishes at $\alpha_k+ \lambda_k$ in each hyperperiod.
This provides a necessary time budget for the execution of the jobs in $S_k$, and eliminates the propagation of timing defects that occur in $S_k$ to the following ETSs, as the execution of $S_k$ cannot exceed $\lambda_k$ regardless of the impact from timing defects.
However, for jobs in $S_k$, this capacity is insufficient to cope with unexpected timing defects within $S_k$, leading to unfinished job executions.

To mitigate the impact of timing defects within an ETS, an additional budget is provided to cope with the execution delay, allowing jobs in the ETS an extra chance to finish executions within their deadlines. However, this could result in a delay to the following ETSs (and the I/O jobs). 
Thus, to enable this additional budget while not jeopardising the real-time guarantee of the system, the following two bounds are computed first, for each $S_k \in S^* \cup S^\neg$:
\begin{itemize}
\item \textit{the highest delay} that $S_k$ can incur, denoted as $\Upsilon_k$; and
\item \textit{the largest slack} between $S_k$ and the next ETS (say $S_{k+1}$), denoted as $\Psi_k$.
\end{itemize}

The first bound implies that a certain amount of additional capacity can be provided to the previous ETS (if required), without causing immediate deadline misses in $S_k$ (see Lemma~\ref{lem:delay}). The second provides a limitation on such capacity that prevents potential deadline misses in the following ETS (\ie, $S_{k+1}$) due to the transitive delay effect from $S_k$.


\begin{lemma}\label{lem:delay}
For a given $S_k \in S^* \cup S^\neg$, the highest delay that $S_k$ can incur is $ \Upsilon_k = \min\{ d_i^j - \theta_i^j - C_i | ~\forall \tau_i^j \in G(S_k)\}$.
\end{lemma}
\begin{proof}
Given a list schedule of I/O jobs in $S_k$ (see Alg.~\ref{alg:sched}, where each job is executed based on an explicit start offset), the worst-case delay effect occurs if all jobs are tightly executed, \ie, without any timing gap between the execution of two jobs. In this case, a delay of $\Upsilon_k$ on $\alpha_k$ can directly impose the same amount of latency on $\theta_i^j$ of each job in $S_k$, leading to deadline misses for jobs with $\theta_i^j+C_i+\Upsilon_k > d_i^j$. Therefore, to provide timing guarantee, $\Upsilon_k \leq \min\{d_i^j - \theta_i^j - C_i\}, \forall \tau_i^j \in G(S_k)$ must hold, hence, the lemma follows. 
\end{proof}


\begin{lemma}\label{lem:slack}
For $S_k$ and $S_{k+1}$, 
the largest slack between $S_k$ and $S_{k+1}$ is $\Psi_k = (\alpha_{k+1} + \min\{\Upsilon_{k+1}, \Psi_{k+1}\}) -\alpha_k-\lambda_k$. For the last $S_k$, $\Psi_k=T^H - \alpha_k-\lambda_k$.
\end{lemma}
\begin{proof}
First, given two consecutive ETSs $S_k$ and $S_{k+1}$, the largest lack between them appears when $S_k$ has the earliest finish (\ie, $\alpha_k + \lambda_k$), whereas $S_{k+1}$ has the latest start.
Then, for $S_{k+1}$, its latest start offset is bounded by both $\Upsilon_{k+1}$ and $\Psi_{k+1}$, because (i) the bound of $\Upsilon_{k+1}$ prevents deadline misses in $S_{k+1}$ (proved in Lemma~\ref{lem:delay}), and (ii) $\Psi_{k+1}$ imposes a bound on the transitive delay of the following ETSs (if they exist) that avoids further deadline misses.
\end{proof}



Based on Lemmas~\ref{lem:delay} and~\ref{lem:slack}, Theorem~\ref{the:adcapacity} describes the additional capacity of $S_k$ (denoted as $\omega_k$) without affecting the timing of I/O jobs in the system.

\begin{theorem} \label{the:adcapacity}
For a given $S_k$, it can execute using an additional capacity of $\omega_k = min\{\Upsilon_{k+1}, \Psi_{k+1} \}$ without jeopardising the timing guarantee of the system.
\end{theorem}
\begin{proof}
Following Lemmas 1 and 2, this theorem can be proved using straightforward counterexamples.
If $\omega_k =  \Psi_{k+1} > \Upsilon_{k+1}$, the delay imposed on $S_{k+1}$ from $S_k$ would cause direct deadline misses of I/O jobs in $S_{k+1}$. This is proved in Lemma~\ref{lem:delay}.
Otherwise, \ie, $\omega_k = \Upsilon_{k+1} > \Psi_{k+1}$, the delay incurred by $S_{k+1}$ can transitively affect the execution of the following ETSs, leading to further deadline misses. This is proved in Lemma~\ref{lem:slack}.  
Therefore, the theorem holds. 
\end{proof}

As with $\Psi_k$, 
the computation of $\omega_k$ starts from the last $S_k \in S^* \cup S^\neg$ with $\omega_k = min\{\Upsilon_{k+1}, (T^H - \alpha_k-\lambda_k )\} $, and calculates $\omega_k$ for each $S_k$ backwards to provide the maximum additional capacity possible when being required. 
With the support of this additional capacity,
jobs in $S_{k}$ can execute until the time instant $\alpha_{k+1}+\omega_{k}$ with an additional budget up to $\Psi_{k}$, \ie, the slack between the finish of $S_{k}$ and the latest start of $S_{k+1}$.
In addition, by taking $\Upsilon_{k+1}$ and $\Psi_{k+1}$ into account when computing $\omega_k$, we prevent potential deadline misses due to the transitive delay effect on the following ETSs, \ie, the start of $S_{k+1}$ is delayed with a bound of $\omega_k$.

With $\omega_k$ obtained, the final capacity of $S_k$ is determined as $\lambda_k = \max\{  \theta_i^j + C_i ~|~ \forall \tau_i^j \in G(S_k)\} - \alpha_k + \omega_k$ based on Equation~\ref{eq:lambda}. This provides an extra execution budget for $S_k$ to cope with unforeseen timing defects, while not endangering the timing requirements of I/O jobs. Based on the ISA constructed in Sec.~\ref{sbsc:ISASupport}, the computed $\alpha_k$, $\lambda$, and $T_k$ of $S_k \in S^* \cup S^\neg$ can be configured in \name{} using instructions {$\text{{\texttt{c.enr}}}, \alpha_k, S_k$} and {$\text{{\texttt{c.set}}}, \lambda_k, S_k$}, respectively (see Tab.~\ref{table:ISA}).

To this end, we constructed \nameS{} that schedules the I/O jobs on ETSs provided by the \name{}. The proposed \nameS{} provides the promised timing-accurate I/O control without timing defects. More importantly, using ETSs with the additional capacity, \nameS{} enhances robustness against unforeseen timing defects, as described in Theorem~\ref{the:robustness}.
\begin{theorem}\label{the:robustness}
\nameS{} can tolerate an interference of $\omega = \min\{ 
\omega_k
~|~ \forall S_k \in S^* \cup S^\neg\}$.
\end{theorem}
\begin{proof}
For a given $S_k \in S^* \cup S^\neg$, it can cope with an additional interference of at most $\omega_k$ without causing a deadline miss in $S_{k+1}$ (if it exists). This is proved in Theorem~\ref{the:adcapacity}. Hence, the system can at least cope with a minimum interference of $\omega$ due to timing defects, while not affecting the timing of I/O jobs of other ETSs with $\omega \leq \min\{\omega_k ~|~ \forall S_k \in  S^* \cup S^\neg\}$. Therefore, the theorem follows.
\end{proof}

Following directly from this theorem, \nameS{} can tolerate an interference of $\omega$ while guaranteeing the timing of each I/O job (including those in the ETS with timing defects) if $\omega < \min\{ 
\min\{\omega_k, \Upsilon_k\}~|~ \forall S_k \in S^* \cup S^\neg\}$. Under this case, the interference caused by timing defects in $S_k$ is not higher than the maximum delay that $S_k$ can incur (\ie, $\Upsilon_k$). Hence, jobs in $S_k$ can meet their deadlines, as proved in Lemma~\ref{lem:delay}.



Finally, it is worth noting that an early ETS that consumes the extra budget (incurs timing defects) can result in less additional budget for later ones, \ie, dynamic use of the additional capacity. 
However, this will not jeopardise the temporal isolation provided by the ETSs, which is effectively bounded by $\Omega_{k}$ for $S_k$ so that the predictability of jobs in later ETSs is not affected. That is, the extended execution of an ETS will not cause any additional deadline misses in later ETSs. 
For instance, if $S_k$ executes until $\alpha_{k+1}+\omega_{k}$, $S_{k+1}$ will have a reduced additional capacity as it must be finished at $\alpha_{k+1}+\lambda_{k+1}$ regardless of the delay. By contrast, if $S_{k}$ finishes execution before $\alpha_{k+1}$, $S_{k+1}$ can still start at $\alpha_{k+1}$ with improved timing accuracy.
This increases the success ratio of the system, where a job can execute and deliver using the additional budget when timing defects occur. This is further justified in Sec.~\ref{sec:results_performance} by experimental results.

\newcommand\BSRTOS{BS$|$RTOS}
\newcommand\BSGPIOCP{BS$|$GPIOCP}
\newcommand\BSRTIOC{BS$|$RTIOC}
\newcommand\BSGA{BS$|$GA}

\section{Evaluation} 
\label{sc:Experiments}
\parlabel{Experimental platform.}
We built the \name\ on a Xilinx VC709 evaluation board. The I/O controller was implemented using Chisel (v.3.4) and connected to a $5 \times 5$ mesh type NoC (Fig.~\ref{fig:concept} (a)). 
As well as the \name{}, the NoC also contained 8/16 open-source RISC-V processors~\cite{asanovic2016rocket}, shared L2 cache (512 KB), external memory (4 GB), and I/O devices.
We instantiated the processors with a 5-staged pipeline and single-width dispatch.
All the hardware elements were synthesised, placed, and routed using Xilinx Vivado (v.2022.2).
The software executing on the processors (OS, drivers, and applications) was compiled using a RISC-V GNU tool-chain (v.2022.11). 
We selected FreeRTOS (v.10.4) as the OS kernel, with the modifications that
were introduced in Sec.~\ref{sbsc:SystemArchitecture}.
To enable comparisons, we introduced 3 Baseline Systems (\textbf{BS}s) on similar hardware, replacing \name\ with different I/O controllers:
\BSRTOS~\cite{FreeRTOS} is a legacy BS with the standard I/O controller, leaving I/O scheduling and management to the RTOS at software level;
\BSGPIOCP{}~\cite{jiang2017gpiocp} and \BSRTIOC{}~\cite{zhao2020timing} were built upon the SOTA real-time I/O controllers with different methodologies (see Sec.~\ref{sc:RelatedWork}).
All systems ran at 100 MHz due to the use of an FPGA as a prototyping platform.

\subsection{Hardware Overhead}
\label{sbsc:HardwareOverhead}
\parlabel{Experimental setup.}
We configured \name~to support 8 ETSs and compared its overhead with a standard Ethernet controller and the real-time controllers used in the \BSRTOS, \BSGPIOCP\ and \BSRTIOC, respectively.
The Ethernet controller was chosen from the Xilinx IP library with default settings, and the GPIOCP and RT-IOC were instantiated to ensure the same capacity with \name.
In addition, to examine the overhead from a system perspective, we compared \name\ against two general-purpose RISC-V processors (Rocket~\cite{asanovic2016rocket} and Boom~\cite{zhao2020sonicboom}).
Rocket was configured using the settings described in the experimental platform. 
Boom had all the features of the Rocket processor, with the extra support of 3-width instruction-level parallelism.
All components were compared using LUTs, registers, and BRAMs.
All components were synthesised and implemented by Xilinx Vivado (v.2022.2) and compared using Look-Up-Tables (LUTs), registers, and BRAMs.
Since these metrics were evaluated using various units, we the results using the Ethernet controller: 4,393 LUTs, 5,113 registers, 2 DSPs and 16KB BRAMs.

\parlabel{Obs 1.} \name~used less hardware than other controllers.
From the SoC's perspective, the overhead was trivial.

This observation is given in Fig.~\ref{fig:Overhead_Small}(a), the implementation of \name~consumed similar hardware compared to the standard Ethernet controller: 108.2\% LUTs, 83.9\% registers, 100.0\% BRAMs.
When compared to the other SOTA real-time controllers, \name\ demonstrated lower overhead across all metrics: GPICOP (92.6\% LUTs, 74.4\% registers, 50\% BRAMs) and RT-IOC (82.3\% LUTs, 72.9\% registers, 50\% BS`RAMs).
The improvement is attributed to the resource-efficient micro-architecture introduced in Sec.~\ref{sc:Design}.
From the SoC's view, \name's overhead is negligible: when compared with the general-purpose Rocket and Boom cores, \name\ only required 20\% and 3\% overhead, respectively.

\begin{figure}[t]
    \centering
    \hspace{-10pt} \includegraphics[width=1\columnwidth]{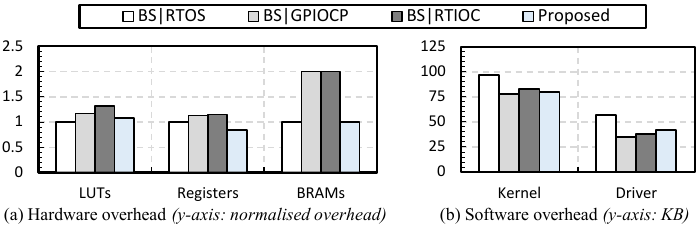}
    \caption{Analysis of hardware and software overheads.}
    \label{fig:Overhead_Small}
\end{figure}

\begin{table}[t]
\centering
\caption{System schedulability under varied $U$.}
\label{tab:sched_res}
\resizebox{.98\columnwidth}{!}{%
\begin{tabular}{c|lllllll}
$U=$ & 0.2  & 0.3  & 0.4  & 0.5  & 0.6  & 0.7  & 0.8  \\
\hline
\multicolumn{1}{c}{} &  \multicolumn{7}{c}{$P_r = 0$,~~$P_e = 0$}\\
\hline
BS$|$GPIOCP &  0.69 &	0.42 &	0.23 &	0.09 &	0.02 &	0.01 &	0.00
\\
BS$|$GA & 0.99	&0.95	&0.84	&0.69&	0.53	&0.38&	0.20
\\
BS$|$RT-IOC      & 0.87 & 0.73 & 0.59 & 0.41 & 0.28 & 0.19 & 0.10 \\
\textbf{Proposed} & \textbf{0.86} & \textbf{0.70} & \textbf{0.55} & \textbf{0.38} & \textbf{0.21} & \textbf{0.11} & \textbf{0.05} \\
\hline
\multicolumn{1}{c}{} & \multicolumn{7}{c}{$P_r = 0.3$,~~$P_e = 0.5$}\\
\hline
BS$|$GPIOCP &  0.59	 & 0.3 &	0.1	 & 0.03	 & 0.00	 & 0.00	 & 0.00  \\
BS$|$GA & 0.74	 &0.45	  &0.22	 & 0.08	 & 0.01	 &0.00	  & 0.00 
 \\
BS$|$RT-IOC      & 0.72 & 0.45 & 0.24 & 0.09 & 0.02 & 0.00 & 0.00 \\
\textbf{Proposed} & \textbf{0.77} & \textbf{0.53} & \textbf{0.33} &\textbf{0.17} & \textbf{0.06} & \textbf{0.02} & \textbf{0.00}
\end{tabular}}
\end{table}

\subsection{Software Overhead}
\label{sbsc:SoftwareOverhead}
Deploying \name\ requires kernel changes and a new driver (Sec.~\ref{sbsc:SystemArchitecture}), hence we examined the software overhead.

\parlabel{Experimental setup.}
We examined the software overhead via run-time memory footprint, with consideration of the OS kernel and the I/O drivers (unit: KB). 
The vanilla kernel used in the examined systems was fully featured with the software I/O manager and essential I/O drivers~\cite{FreeRTOS}. 
The memory size tool was RISC-V GNU tool-chain.

\parlabel{Obs 2.} \name\ consumed less software than \BSRTOS. Its overhead was similar to \BSGPIOCP\ and \BSRTIOC.

In Fig.~\ref{fig:Overhead_Small}(b), the OS kernel and drivers required by \name\ consumed 32 KB (20.8\%) less memory than the software solution due to its hardware-implemented I/O management. 
Compared to other solutions, the memory usage of  \name~was similar: 107.9\% (\BSGPIOCP) and 98.7\% (\BSRTIOC).

\subsection{I/O-level Timing Performance} \label{sec:results_performance}
This section evaluates the timing performance of the proposed I/O control method against the \textbf{BS}s and the Genetic Algorithm described in~\cite{zhao2020timing} (\BSGA) that optimises the timing performance of I/O tasks. As with~\cite{zhao2020timing}, the timing accuracy was measured as the \textit{percentage of exact-accurate jobs} and \textit{operation quality} based on $\Phi_i$ normalised by the maximum quality achievable (\ie, $\sum_{\tau_i^j \in \Gamma^H} V_{max}$).


\parlabel{Experimental setup.}
The number of I/O tasks was set to $n=[4,16]$, and the utilisation of tasks was generated by the UUniFast algorithm~\cite{bini2005measuring} with a total system utilisation given by $U =0.05 \times n$.
For each $\tau_i$, $T_i$ was generated randomly with a uniform distribution with a hyper-period of $1440ms$, and $D_i=T_i$.
The $\Phi_i$ of a task was randomly chosen from examples in Fig.~\ref{fig:vc}, with $V_{max}$ and $V_{min}$ randomly decided from $[1, 100]$ and $[0, V_{max}]$, respectively.
In addition, two control variables were defined for timing defects: $P_r$ controlled the percentage of jobs that incurred timing defects and $P_e$ gave the interference incurred by such a job, quantified as $C_i \times (1 + P_e)$. 
For each system configuration, 1,000 systems were randomly generated and evaluated. 

\begin{figure}[t]
\centering
\hspace{-15pt} \includegraphics[width=1.035\columnwidth]{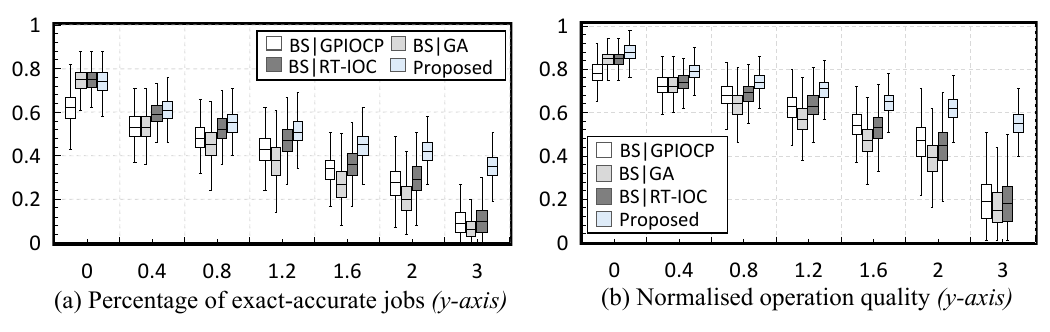}
\caption{Timing accuracy with $U=0.6$, $P_r=0.3$ \small{\emph{($x$-axis: $P_e$)}}. }
\label{fig:accuracy_td}
\end{figure}

\begin{figure}[t]
\centering
\hspace{-15pt} \includegraphics[width=1.035\columnwidth]{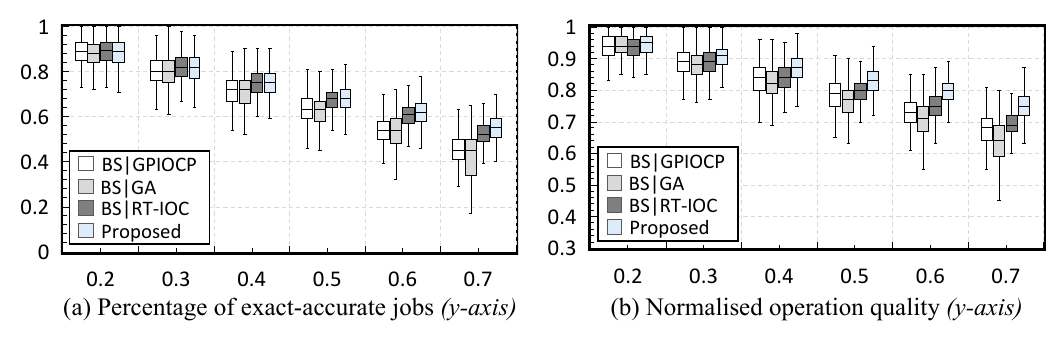}
\caption{Timing accuracy with $P_r = 0.3$, $P_e = 0.5$ \small{\emph{($x$-axis: $U$)}}.}
\label{fig:accuracy}
\end{figure}

\parlabel{Obs 3.}  \name{} showed higher schedulability and timing accuracy compared to SOTA methods with timing defects.

This observation is seen in Table~\ref{tab:sched_res} and Fig.~\ref{fig:accuracy}, showing the resulting system schedulability and timing accuracy under varied $U$, respectively. 
Table~\ref{tab:sched_res} shows that the \BSGPIOCP{} was constantly outperformed by other methods due to its simple FIFO-based scheduling strategy~\cite{zhao2020timing}.
In addition, we noticed that although \name{} shows a slightly lower schedulability than \BSRTIOC{} without timing defects, it outperforms all SOTA methods when timing defects occur. This is because the \BSRTIOC{} produces the schedule by a bin-packing approach, without considering the tightness of deadlines of I/O jobs, hence, is more fragile to timing defects. 
\BSGA{} tends to produce a tight schedule where jobs are executed close to each other. Thus, although \BSGA{} has the highest schedulability when $P_r=P_e=0$, it showed the most pronounced fall with timing defects.
By contrast, the ETS-based schedule (\ie, \nameS{}) that explicitly executes the job with the earliest $d_i^j$ first (line 16 in Alg.~\ref{alg:sched}), hence, is less likely to be affected by timing defects compared to SOTA methods.

\begin{figure*}[t]
\centering
\hspace{-.02\textwidth}
\subfigure[8-core systems \emph{($x$-axis: U; $y$-axis: SR)}.] {\label{fig:successratio}
\includegraphics[width=.30\textwidth]{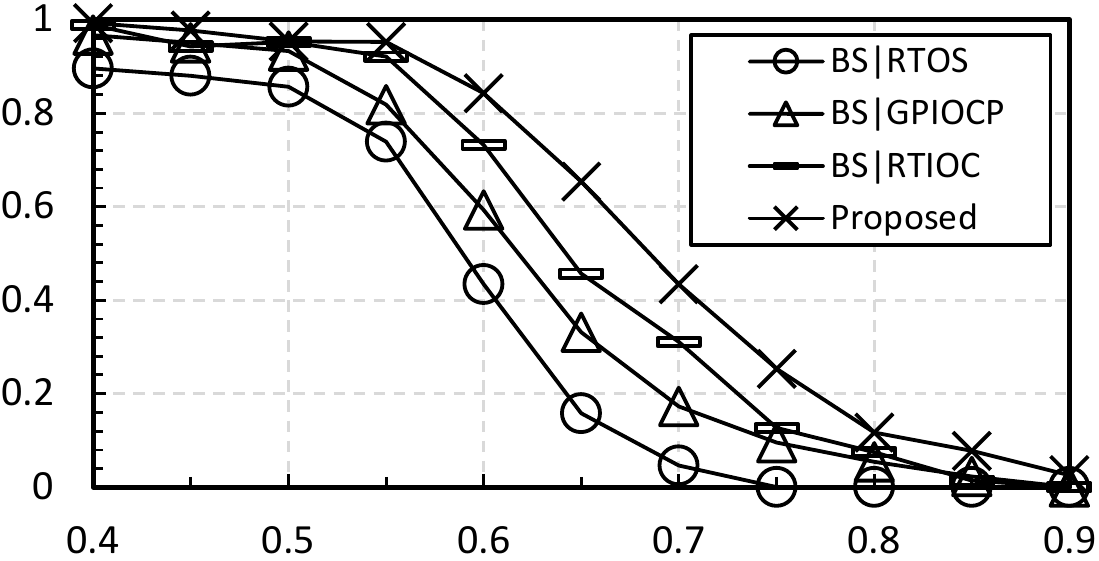}}
\hspace{.02\textwidth}
\centering
\subfigure[32-core systems \emph{($x$-axis: U; $y$-axis: SR)}.]{\label{fig:SR8}
\includegraphics[width=.30\textwidth]{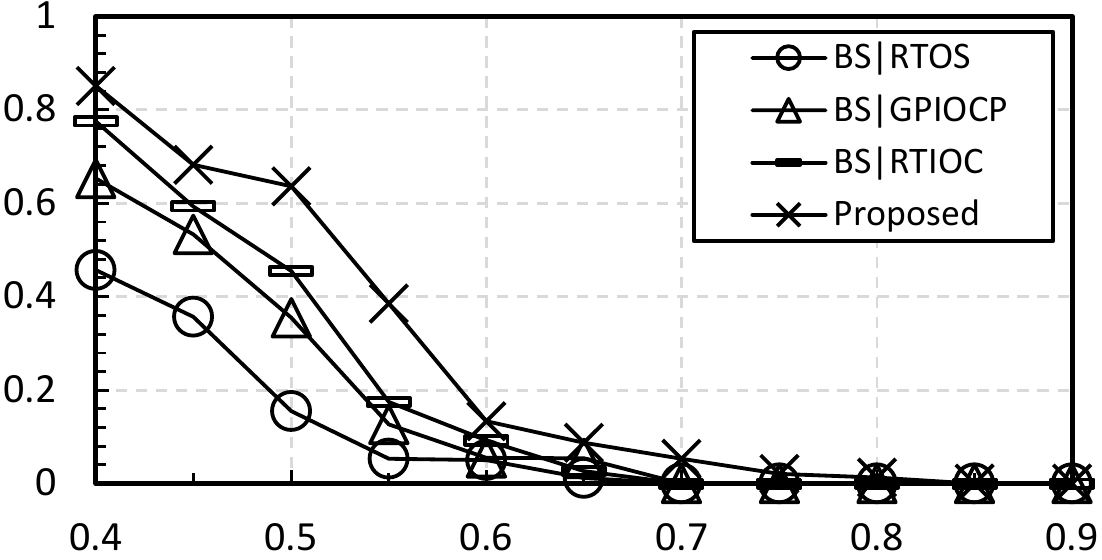}}
\hspace{.02\textwidth}
\centering
\subfigure[{I/O-QoS \emph{($y$-axis: normalised I/O-QoS).}}] {\label{fig:correctnessratio}
\includegraphics[width=.30\textwidth]{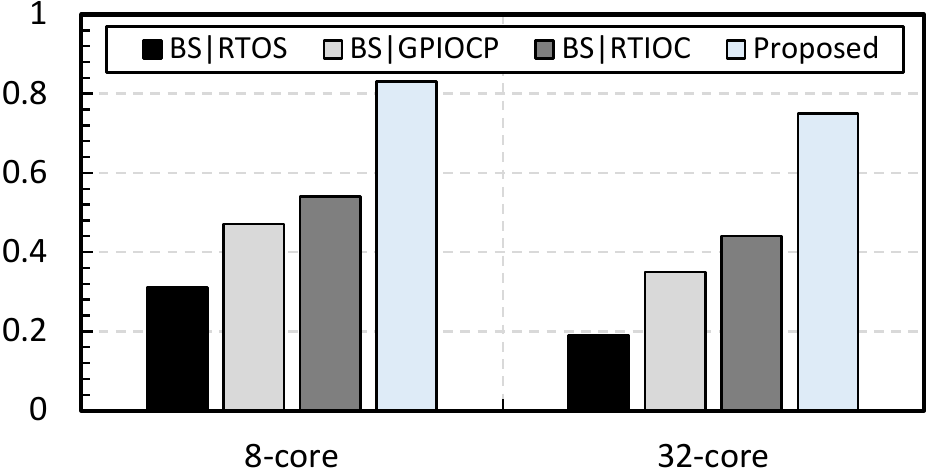}}
\caption{Case study, evaluating system-level real-time performance (in (a) and (b)) and I/O quality (in (c)).}
\label{fig:exp_case}
\end{figure*}

In addition, as shown in Fig.~\ref{fig:accuracy}, 
\name{} demonstrated both a higher percentage of exact-accurate jobs and the strongest overall quality in the general case,
especially when the system was heavy-loaded ($U \geq 0.6$).
This is because, with the ETSs applied, \name{} effectively mitigates the impact of timing defects on the start offset of the I/O jobs, leading to higher timing accuracy compared to the SOTA methods. By contrast, the timing accuracy of the SOTA methods was significantly undermined by timing defects.
For instance, \BSGA{} (the method with the highest quality in~\cite{zhao2020timing}) was outperformed by both \BSRTIOC{} and \BSGPIOCP{} as shown in Fig.~\ref{fig:accuracy} in the general case. This observation revealed the impact of timing defects on existing timing-accurate I/O schedules, and justified the necessity for developing robust I/O control methods.

\parlabel{Obs 4.} \name{} improves robustness compared to SOTA methods when impact of timing defects becomes significant.

This observation is seen in both Table~\ref{tab:accept} and Fig.~\ref{fig:accuracy_td}, which report the acceptance ratio (the percentage of jobs that meet their deadlines) and the timing accuracy of systems that are deemed schedulable by the competing methods, with an increasing impact (\ie, $P_e$) of timing defects.
From the results, \name{} showed both a higher acceptance ratio and timing accuracy compared to SOTA methods. For instance, \name{} outperformed \BSRTIOC\ by 22.61\% and 42.32\% on average (2.18x and 3.06x when $P_e = 3$), in terms of acceptance ratio and operation quality, respectively.
One notable reason is that, with the additional budget assigned for the ETSs, \name{} allows more jobs to finish while providing the necessary temporal isolation for jobs in the following ETSs. 
This justified the dynamic use of the additional budget (see Sec.~\ref{sbsc:ets_config}).
By contrast, the SOTA methods do not take unexpected interference from timing defects into account. Thus, these methods are highly sensitive to timing defects, which significantly affect their pre-planned schedule and resulting timing performance, especially when $P_e \geq 1.2$ as shown in Fig.~\ref{fig:accuracy_td}. 


\begin{table}[t]
\centering
\caption{Acceptance ratio with $U=0.6$ and $P_r=0.3$.}
\label{tab:accept}
\resizebox{.98\columnwidth}{!}{%
\begin{tabular}{c|ccccccc}
$P_e=$ & $0$	&$0.4$		&$0.8$&$1.2$&$1.6$&$2.0$ &3.0\\
\hline
BS$|$GPIOCP	&1.00	&0.99	&0.97	&0.94	&0.88	&0.8	&0.43

\\
BS$|$GA	&1.00	&0.99	&0.95	&0.91	&0.82	&0.72	&0.38

\\
BS$|$RT-IOC      & 1.00	&0.98	&0.96	&0.92	&0.84	&0.75	&0.40\\
\textbf{Proposed} & \textbf{1.00}	& \textbf{0.99}	& \textbf{0.98}	& \textbf{0.96}	& \textbf{0.93}	& \textbf{0.92} & \textbf{0.87}
\\
\end{tabular}}
\end{table}

\begin{figure}[t]
    \centering
    \includegraphics[width=1\columnwidth]{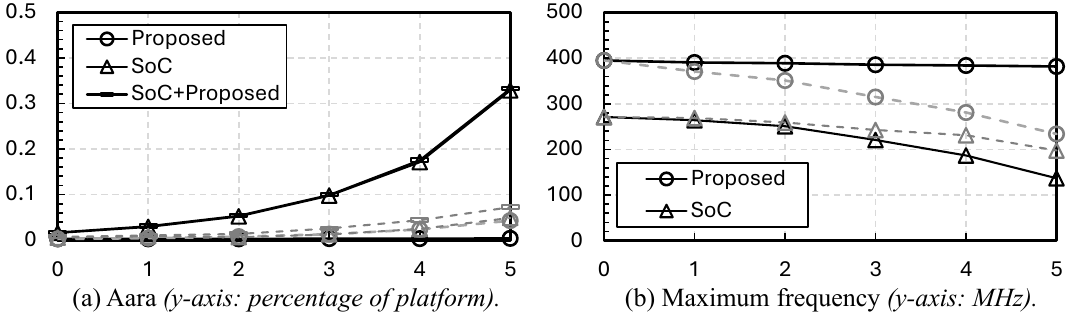}
    \caption{Scalability \small{\emph{($x$-axis: $\eta^{\text{core}}$ for black lines, $\eta^{\text{io}}$ for grey lines)}}.}
    \label{fig:scalability}
\end{figure}


\subsection{System-level Real-time Performance}
\label{sc:CaseStudy}


We evaluated the system-level real-time performance of the systems using an automotive case study. Specifically, we configured all systems with either 8 or 32 cores and executed two distinct sets of tasks: (i) 10 automotive safety tasks selected from the Renesas automotive use case database, such as CRC and RSA32; and (ii) 10 function tasks from the EEMBC benchmark, including FFT and speed calculation. 
We slightly modified these tasks to generate raw data externally, which was then transmitted to the system under test via an Ethernet connection (1 Gbps). 
The results were subsequently transmitted back using two FlexRay channels (10 Mbps).
The value of $\Phi_i$ for each I/O task was randomly chosen from the examples provided in Fig.~\ref{fig:vc}, with $V_{\text{max}}$ and $V_{\text{min}}$ set within the ranges of $[1, 100]$ and $[0, V_{\text{max}}]$, respectively. Each task was associated with a predefined period and an implicit deadline.
Additionally, we introduced synthetic workloads to potentially adjust the overall system utilisation. These workloads, akin to the function tasks, were derived from the EEMBC benchmark. However, since the execution time of a task can be influenced by various factors, the addition of synthetic workloads provided only a \emph{target utilisation} ($U$) for the system.

\parlabel{Experimental setup.}
We executed the systems 200 times under varying target utilisation $[50\%, 100\%]$, at intervals of 5\%.
Each experimental trial lasted 200 seconds, with $[10, 20]$ hardware faults randomly injected to the I/O controllers, delaying the I/Os' execution.
For a fair comparison,  we ensured that both the input data and the injected faults were identical across all tested systems.
We evaluated the systems using Success Ratio (SR) and Quality of I/O Services (QoS-I/O). 
The SR recorded the percentage of trials that were executed successfully (\ie, without deadline miss of any safety and function task), under a specified target utilisation. 
The QoS-I/O reports the average quality among all I/O tasks.

\parlabel{Obs. 5}. Deploying \name\ increased overall real-time performance than baseline systems with improved I/O quality.

As shown in Fig.~\ref{fig:exp_case}(a), (b), and (c), when the examined systems were configured using the same settings, our proposed system consistently outperformed the BSs in terms of SR and I/O-QoS.
The reasons for these improvements are two-fold:
(i) the approach of hardware/algorithm co-design enhances the efficiency of I/O management (see Sec.\ref{sc:Design});
(ii) the micro-architecture of \name~isolates ETSs and prevents the propagation of timing defects -- aligned to Obs. 3.

\parlabel{Obs. 6}.
With the increased complexity of the SoCs (\ie, the core numbers), all systems exhibited a decline in SR.

This pattern is given by the comparison between Fig.~\ref{fig:exp_case}(a) and Fig.~\ref{fig:exp_case}(b).
The reduction in SR stems from the additional on-chip interference and resource contention introduced by the growing number of processors and tasks - characteristic challenges in all multi/many-core systems, particularly when multiple tasks rely on the I/Os.
Even so, in 32-core systems, \name\ consistently achieved the highest SRs, highlighting the effectiveness and applicability of \name\ in complicated many-core SoCs. 
To further mitigate such decline, it is necessary to consider the NoC contentions and interference to facilitate an end-to-end optimisation~\cite{jiang2023nprc,restuccia2019your,pagani2020bandwidth}.

\subsection{System-level Scalability}
\label{sc:Scalability}
We acknowledge that the design scalability impacts its feasibility; hence, we performed a scalability analysis.

\parlabel{Experimental setup.} 
The same method described in Sec.\ref{sbsc:HardwareOverhead} was adopted to implement an SoC with a scaling number of processor cores and I/Os.
In addition, we introduced scaling factors: $\eta^{core}$ and $\eta^{io}$ to control the number ($2^\eta$) of cores and I/Os in the system.
As the number of cores increases, the associated I/O pool and L-SEs in a \name\ also expand to buffer and schedule more tasks. 
Conversely, an increase in the number of I/Os requires a corresponding increase in the number of {\name}s for additional device management.

First, we compared the scalability of area consumption between the legacy SoC (built upon standard ETH controllers) and the \name-based SoCs. 
The area consumption was normalised by the overall area of the platform.
We then examined the maximum frequency of the \name~and the SoC using varying $\eta^{core}$ and $\eta^{io}$.

\parlabel{Obs. 7.}  \name's area was almost  unaffected by $\eta^{\text{core}}$ and linearly scaled by $\eta^{\text{io}}$. 
Also, implementing \name\ in multi-/many-core SoCs did not impact the maximum performance.

As seen in Fig.\ref{fig:scalability}(a), when systems were scaled with $\eta^{\text{core}}$, the area consumption of \name\ was nearly constant.
When scaled with $\eta^{\text{io}}$, the area of \name\ increased linearly, demonstrating the micro-architecture's hardware scalability. 
Furthermore, as illustrated in Fig.~\ref{fig:scalability}(b), in all examined cases, \name's maximum frequency was always higher than the baseline SoCs.
This indicates that \name\ did not become a critical path of the entire SoC design, and therefore did not limit the maximum performance of the system.

\subsection{Summary}
Based on the observations given in all experiments, \name{} showed robustness by coping with timing defects under both synthesised systems and a real-world case study, improving timing predictability and accuracy compared to SOTA methods, without consuming significantly more resources than SOTA I/O controllers. In addition, the scalability of \name{} is demonstrated using real-world implementations, validating both the applicability and the effectiveness of the proposed solution on large-scale industrial applications.


\section{Related Work}
\label{sc:RelatedWork}
Research into real-time I/O design and verification mehtods in multi-/many-core systems has been investigated from the perspective of both the software and hardware.

\parlabel{Software-based I/O scheduling.}
Conventional safety-critical systems rely on an OS or a hypervisor to manage the I/Os~\cite{burns2001real}.
Hence, software approaches aimed at satisfying I/O timing demands have typically modified the OS kernel~\cite{kim2014integrated,betti2011real,kim2018supporting} or the hypervisor~\cite{missimer2016mixed,west2016virtualized} to integrate I/O-aware task scheduling and mapping strategies~\cite{casini2021latency}, avoiding unnecessary contentions and trying to execute I/O operations as accurately as possible. 
A considerable amount of work has contributed to this research domain by developing different task mapping and scheduling methods, \eg, reserving I/O bandwidths for software tasks according to their workload~\cite{abdallah2016contention,kim2018supporting}, and applying a time utility function to schedule tasks at or close to their ideal offsets~\cite{wang2004time,li2004adaptive}.
However, as explained in Sec.~\ref{sc:Introduction}, software solutions are becoming difficult to apply to modern systems, due to the never-ending increase in hardware and architectural complexity, \eg, many-core systems with mesh Network-on-Chips (NoCs).
Such complexity introduces significant uncertainty to I/O transmissions, \eg, communication delays and resource contentions through the system, making accurate control extremely tough to achieve from the software.

\parlabel{I/O-aware NoCs.}
At the NoC level, a commonly used method to improve real-time performance of I/O transactions is the duplication of communication channels.
For instance, such duplicated channels were particularly deployed for different types of I/O devices, \eg, cameras~\cite{parten2017multi} and display devices~\cite{chandra2014multi}.
Similarly, Jiang~\etal~\cite{jiang2023nprc,jiang2023many} introduced a bypass channel, termed ``I/O ring", facilitating more efficient data routing (one cycle level) between processors and all types of device. 
Distinct from merely duplicating communication channels, bandwidth reservation has also been studied to ensure service for particular tasks or transactions. 
This involves the interconnect assigning priorities to each hardware element and allocating bandwidth accordingly. 
For example, dedicated AXI controllers reserve bandwidth for high-priority devices~\cite{restuccia2019your, hebbache2018shedding, pagani2020bandwidth, jiang2022axi}. 
However, these methods solely focused on schedulability and average throughput of I/Os, without any attention paid to the key factors of timing accuracy and robustness.


\parlabel{Hardware-assisted I/O management.}
At the I/O level, dedicated assists have been developed and manufactured by different semiconductor vendors. For example, a programmable real-time unit has been developed by TI~\cite{PRUWeb} and a time processor unit by NXP~\cite{TPUWeb}.
These controllers are physically connected to I/O devices and handle I/O operations directly at the hardware level.
Since I/O management is deployed close to I/O devices, it effectively avoids I/O transmission uncertainty, and hence, provides the potential for achieving I/O timing predictability and accuracy. 
Leveraging the hardware support, Jiang~\etal~\cite{jiang2017gpiocp} and Zhao~\etal~\cite{zhao2020timing} presented different configuration methods that ordered the execution of I/O operations, specifying the start time of each I/O task (\ie a series of I/O operations). 
Similarly, Guerra and Fohler~\cite{guerrra2008gravitational} introduced a gravitational pendulum model that assigned higher ``gravity" to I/O tasks with higher importance to improve their unity.
These methods achieved a certain amount of timing predictability and accuracy for I/O control in their system models~\cite{guerrra2008gravitational} and in work using their assumptions~\cite{jiang2017gpiocp,zhao2020timing}.
However, as discussed in Sec.~\ref{sc:Introduction}, these approaches are fragile and would be entirely undermined by the interference of many timing defects.

\section{Conclusion} \label{sc:Conclusion}

This paper presents a real-time I/O controller named~\name{} using a hardware/software co-design approach. \name{} supports configurable ETSs and a two-level scheduler to establish the scheduling infrastructure that can prevent the propagation of timing defects between ETSs. Based on~\name{}, an ETS-based I/O scheduling method named~\nameS{} was constructed to mitigate the impact of timing defects, hence, improving both robustness and timing accuracy. Experimental results demonstrated that~\name{} outperforms SOTA methods without consuming significantly more hardware resources. 
In future work, we plan to tape out the \name{} with safety-critical micro-controllers in a real chip to examine its energy efficiency under different application scenarios.

\parlabel{Lessons learned.}
Different from the conventional solutions that attempted to manage I/O timing defects solely through hardware design (e.g., physical isolation~\cite{herber2014spatial}) or software techniques (e.g., virtualisation~\cite{herber2014spatial}), this work shows that using a hardware/algorithm co-designed approach, timing accuracy and robustness can be effectively improved with light overhead.
The constructed \name{} and \nameS{} 
provide key insights and effective means that foster a collaborative environment between the hardware and the schedule for real-time systems that neither discipline could achieve independently.

\section{Acknowledgement}
The authors would like to thank the anonymous reviewers for their insightful and helpful feedback. 
This work is supported by the National Natural Science Foundation of China (NSFC) under Grant 62472086, Grant 62072478, Grant 62302533, and Grant 62472086, Natural Science Foundation of Jiangsu Province under Grants BK20243042, Guangdong Basic and Applied Basic Research Foundation under Grant 2024A1515010240, Guangzhou Fundamental Research Funds under Grant SL2023A04J00996, the Start-up Research Fund of Southeast University under Grant No. RF1028624005, and the Fundamental Research Funds for the Central Universities. 
This paper was written during the birth of Zhe Jiang's first baby -- Zhe Jiang would like to thank his wife (Yanting Dai) and daughter (Anyu Jiang) for the joy they have brought.


\begin{thebibliography}{10}
\providecommand{\url}[1]{#1}
\csname url@samestyle\endcsname
\providecommand{\newblock}{\relax}
\providecommand{\bibinfo}[2]{#2}
\providecommand{\BIBentrySTDinterwordspacing}{\spaceskip=0pt\relax}
\providecommand{\BIBentryALTinterwordstretchfactor}{4}
\providecommand{\BIBentryALTinterwordspacing}{\spaceskip=\fontdimen2\font plus
\BIBentryALTinterwordstretchfactor\fontdimen3\font minus \fontdimen4\font\relax}
\providecommand{\BIBforeignlanguage}[2]{{%
\expandafter\ifx\csname l@#1\endcsname\relax
\typeout{** WARNING: IEEEtran.bst: No hyphenation pattern has been}%
\typeout{** loaded for the language `#1'. Using the pattern for}%
\typeout{** the default language instead.}%
\else
\language=\csname l@#1\endcsname
\fi
#2}}
\providecommand{\BIBdecl}{\relax}
\BIBdecl

\bibitem{iso201126262}
I.~ISO, ``26262: Road vehicles-functional safety,'' 2018.

\bibitem{hennessy2011computer}
J.~L. Hennessy, \emph{Computer architecture: a quantitative approach}, 2011.

\bibitem{jiang2018bluevisor}
Z.~Jiang, N.~C. Audsley, and P.~Dong, ``Bluevisor: A scalable real-time hardware hypervisor for many-core embedded systems,'' in \emph{2018 IEEE Real-Time and Embedded Technology and Applications Symposium (RTAS)}.\hskip 1em plus 0.5em minus 0.4em\relax IEEE, 2018, pp. 75--84.

\bibitem{borgioli2022virtualization}
N.~Borgioli, M.~Zini, D.~Casini, G.~Cicero, A.~Biondi, and G.~Buttazzo, ``An {I/O} virtualization framework with i/o-related memory contention control for real-time systems,'' \emph{IEEE Transactions on Computer-Aided Design of Integrated Circuits and Systems}, vol.~41, no.~11, pp. 4469--4480, 2022.

\bibitem{chai2019six}
R.~Chai, A.~Tsourdos, A.~Savvaris, S.~Chai, Y.~Xia, and C.~P. Chen, ``{Six-DOF} spacecraft optimal trajectory planning and real-time attitude control: a deep neural network-based approach,'' \emph{IEEE transactions on neural networks and learning systems}, vol.~31, no.~11, pp. 5005--5013, 2019.

\bibitem{mossinger2010software}
J.~Mossinger, ``Software in automotive systems,'' \emph{IEEE software}, vol.~27, no.~2, p.~92, 2010.

\bibitem{zhao2020timing}
S.~Zhao, Z.~Jiang, X.~Dai, I.~Bate, I.~Habli, and W.~Chang, ``Timing-accurate general-purpose {I/O} for multi-and many-core systems: Scheduling and hardware support,'' in \emph{2020 57th ACM/IEEE Design Automation Conference (DAC)}.\hskip 1em plus 0.5em minus 0.4em\relax IEEE, 2020, pp. 1--6.

\bibitem{jiang2017gpiocp}
Z.~Jiang and N.~C. Audsley, ``{GPIOCP}: Timing-accurate general purpose {I/O} controller for many-core real-time systems,'' in \emph{Design, Automation \& Test in Europe Conference \& Exhibition (DATE), 2017}.\hskip 1em plus 0.5em minus 0.4em\relax IEEE, 2017, pp. 806--811.

\bibitem{davis2007robust}
R.~I. Davis and A.~Burns, ``Robust priority assignment for fixed priority real-time systems,'' in \emph{Proc. of RTSS}, 2007.

\bibitem{burns2001real}
A.~Burns and A.~J. Wellings, \emph{Real-time systems and programming languages: Ada 95, real-time Java, and real-time POSIX}, 2001.

\bibitem{davis2011survey}
R.~I. Davis and A.~Burns, ``A survey of hard real-time scheduling for multiprocessor systems,'' \emph{ACM computing surveys (CSUR)}, vol.~43, no.~4, pp. 1--44, 2011.

\bibitem{rierson2017developing}
L.~Rierson, \emph{Developing safety-critical software: a practical guide for aviation software and DO-178C compliance}.\hskip 1em plus 0.5em minus 0.4em\relax CRC Press, 2017.

\bibitem{ballard2021machine}
Z.~Ballard, C.~Brown, A.~M. Madni, and A.~Ozcan, ``Machine learning and computation-enabled intelligent sensor design,'' \emph{Nature Machine Intelligence}, vol.~3, no.~7, pp. 556--565, 2021.

\bibitem{burns2018robust}
A.~Burns, R.~I. Davis, S.~Baruah, and I.~Bate, ``Robust mixed-criticality systems,'' \emph{IEEE Transactions on Computers}, vol.~67, no.~10, pp. 1478--1491, 2018.

\bibitem{abdallah2016contention}
L.~Abdallah, M.~Jan, J.~Ermont, and C.~Fraboul, ``{I/O} contention aware mapping of multi-criticalities real-time applications over many-core architectures,'' in \emph{22nd IEEE Real-Time and embedded Technology and Applications symposium (RTAS 2016)}, 2016, pp. pp--25.

\bibitem{kim2014integrated}
J.-E. Kim, M.-K. Yoon, R.~Bradford, and L.~Sha, ``Integrated modular avionics {(IMA)} partition scheduling with conflict-free {i/o} for multicore avionics systems,'' in \emph{2014 IEEE 38th Annual Computer Software and Applications Conference}.\hskip 1em plus 0.5em minus 0.4em\relax IEEE, 2014, pp. 321--331.

\bibitem{betti2011real}
E.~Betti, S.~Bak, R.~Pellizzoni, M.~Caccamo, and L.~Sha, ``Real-time {I/O} management system with cots peripherals,'' \emph{IEEE Transactions on Computers}, vol.~62, no.~1, pp. 45--58, 2011.

\bibitem{kim2018supporting}
N.~Kim, S.~Tang, N.~Otterness, J.~H. Anderson, F.~D. Smith, and D.~E. Porter, ``Supporting {I/O} and {IPC} via fine-grained os isolation for mixed-criticality real-time tasks,'' in \emph{Proceedings of the 26th international conference on real-time networks and systems}, 2018, pp. 191--201.

\bibitem{brandenburg2022multiprocessor}
B.~B. Brandenburg, ``Multiprocessor real-time locking protocols,'' in \emph{Handbook of Real-Time Computing}.\hskip 1em plus 0.5em minus 0.4em\relax Springer, 2022, pp. 347--446.

\bibitem{zhao2018thesis}
S.~Zhao, ``{A FIFO Spin-based Resource Control Framework for Symmetric Multiprocessing},'' Ph.D. dissertation, The University of York, 2018, \url{http://etheses.whiterose.ac.uk/21014/}.

\bibitem{audsley1993applying}
N.~Audsley, A.~Burns, M.~Richardson, K.~Tindell, and A.~J. Wellings, ``Applying new scheduling theory to static priority pre-emptive scheduling,'' \emph{Software engineering journal}, vol.~8, no.~5, pp. 284--292, 1993.

\bibitem{TPUWeb}
\url{https://www.nxp.com/eTPU}, {TPU}, accessed Nov. 17, 2022.

\bibitem{PRUWeb}
\url{http://www.ti.com/tool/pru-swpkg}, {PRU}, accessed Nov. 17, 2022.

\bibitem{lee2012realizing}
J.~Lee, S.~Xi, S.~Chen, L.~T. Phan, C.~Gill, I.~Lee, C.~Lu, and O.~Sokolsky, ``Realizing compositional scheduling through virtualization,'' in \emph{2012 IEEE 18th Real Time and Embedded Technology and Applications Symposium}.\hskip 1em plus 0.5em minus 0.4em\relax IEEE, 2012, pp. 13--22.

\bibitem{shin2003periodic}
I.~Shin and I.~Lee, ``Periodic resource model for compositional real-time guarantees,'' in \emph{RTSS 2003. 24th IEEE Real-Time Systems Symposium, 2003}.\hskip 1em plus 0.5em minus 0.4em\relax IEEE, 2003, pp. 2--13.

\bibitem{shin2004compositional}
------, ``Compositional real-time scheduling framework,'' in \emph{25th IEEE International Real-Time Systems Symposium}.\hskip 1em plus 0.5em minus 0.4em\relax IEEE, 2004, pp. 57--67.

\bibitem{davis2009robust}
R.~I. Davis and A.~Burns, ``Robust priority assignment for messages on controller area network {(CAN)},'' \emph{Real-Time Systems}, vol.~41, pp. 152--180, 2009.

\bibitem{guerrra2008gravitational}
R.~Guerrra and G.~Fohler, ``A gravitational task model for target sensitive real-time applications,'' in \emph{2008 Euromicro Conference on Real-Time Systems}.\hskip 1em plus 0.5em minus 0.4em\relax IEEE, 2008, pp. 309--317.

\bibitem{guerra2009gravitational}
R.~Guerra and G.~Fohler, ``A gravitational task model with arbitrary anchor points for target sensitive real-time applications,'' \emph{Real-Time Systems}, vol.~43, pp. 93--115, 2009.

\bibitem{liu2021real}
S.~Liu, B.~Yu, N.~Guan, Z.~Dong, and B.~Akesson, ``Real-time scheduling and analysis of an autonomous driving system,'' \emph{in Proceedings of RTSS 2021 Industrial Challenge Problem}, pp. 1--47, 2021.

\bibitem{mei2017real}
H.~Mei, ``Real-time stream processing in embedded systems,'' Ph.D. dissertation, University of York, 2017.

\bibitem{chen2019timing}
P.~Chen, W.~Liu, X.~Jiang, Q.~He, and N.~Guan, ``Timing-anomaly free dynamic scheduling of conditional {DAG} tasks on multi-core systems,'' \emph{ACM Transactions on Embedded Computing Systems (TECS)}, vol.~18, no.~5s, pp. 1--19, 2019.

\bibitem{asanovic2016rocket}
K.~Asanovic, R.~Avizienis, J.~Bachrach, S.~Beamer, D.~Biancolin, C.~Celio, H.~Cook, D.~Dabbelt, J.~Hauser, A.~Izraelevitz \emph{et~al.}, ``The rocket chip generator,'' \emph{EECS Department, University of California, Berkeley, Tech. Rep. UCB/EECS-2016-17}, vol.~4, pp. 6--2, 2016.

\bibitem{FreeRTOS}
\url{https://www.freertos.org}, FreeRTOS, accessed July. 15, 2023.

\bibitem{zhao2020sonicboom}
J.~Zhao, B.~Korpan, A.~Gonzalez, and K.~Asanovic, ``{Sonicboom}: The 3rd generation berkeley out-of-order machine,'' in \emph{Fourth Workshop on Computer Architecture Research with RISC-V}, vol.~5, 2020, pp. 1--7.

\bibitem{bini2005measuring}
E.~Bini and G.~C. Buttazzo, ``Measuring the performance of schedulability tests,'' \emph{Real-time systems}, vol.~30, no.~1, pp. 129--154, 2005.

\bibitem{jiang2023nprc}
Z.~Jiang, X.~Dai, R.~Wei, I.~Gray, Z.~Gu, Q.~Zhao, and S.~Zhao, ``{NPRC-I/O}: A {NoC}-based real-time {I/O} system with reduced contention and enhanced predictability,'' \emph{IEEE Transactions on Computer-Aided Design of Integrated Circuits and Systems}, 2023.

\bibitem{restuccia2019your}
F.~Restuccia, M.~Pagani, A.~Biondi, M.~Marinoni, and G.~Buttazzo, ``Is your bus arbiter really fair? restoring fairness in axi interconnects for fpga socs,'' \emph{ACM Transactions on Embedded Computing Systems (TECS)}, vol.~18, no.~5s, pp. 1--22, 2019.

\bibitem{pagani2020bandwidth}
M.~Pagani, E.~Rossi, A.~Biondi, M.~Marinoni, G.~Lipari, and G.~Buttazzo, ``A bandwidth reservation mechanism for {AXI-based} hardware accelerators on {FPGAs},'' in \emph{31st Euromicro Conference on Real-Time Systems (ECRTS 2019)}, 2019.

\bibitem{missimer2016mixed}
E.~Missimer, K.~Missimer, and R.~West, ``Mixed-criticality scheduling with {I/O},'' in \emph{2016 28th Euromicro Conference on Real-Time Systems (ECRTS)}.\hskip 1em plus 0.5em minus 0.4em\relax IEEE, 2016, pp. 120--130.

\bibitem{west2016virtualized}
R.~West, Y.~Li, E.~Missimer, and M.~Danish, ``A virtualized separation kernel for mixed-criticality systems,'' \emph{ACM Transactions on Computer Systems (TOCS)}, vol.~34, no.~3, pp. 1--41, 2016.

\bibitem{casini2021latency}
D.~Casini, A.~Biondi, G.~Cicero, and G.~Buttazzo, ``Latency analysis of {I/O} virtualization techniques in hypervisor-based real-time systems,'' in \emph{2021 IEEE 27th Real-Time and Embedded Technology and Applications Symposium (RTAS)}.\hskip 1em plus 0.5em minus 0.4em\relax IEEE, 2021, pp. 306--319.

\bibitem{wang2004time}
J.~Wang and B.~Ravindran, ``Time-utility function-driven switched ethernet: Packet scheduling algorithm, implementation, and feasibility analysis,'' \emph{IEEE Transactions on Parallel and Distributed Systems}, vol.~15, no.~2, pp. 119--133, 2004.

\bibitem{li2004adaptive}
P.~Li, B.~Ravindran, and E.~D. Jensen, ``Adaptive time-critical resource management using time/utility functions: Past, present, and future,'' in \emph{Proceedings of the 28th Annual International Computer Software and Applications Conference}, vol.~2.\hskip 1em plus 0.5em minus 0.4em\relax IEEE, 2004, pp. 12--13.

\bibitem{parten2017multi}
B.~Parten, Y.~Lee, B.~Quach, L.~Myers, W.~Ray, and W.~Maung, ``Multi-channel peripheral interconnect supporting simultaneous video and bus protocols,'' 2017, {US} Patent.

\bibitem{chandra2014multi}
P.~R. Chandra, K.~C. Kahn, E.~Galil, E.~Kugman, N.~Zolotov, V.~Yudovich, Y.~Dishon, and E.~Bagelman, ``Multi-protocol tunneling over an {I/O} interconnect,'' 2014, {US} Patent.

\bibitem{jiang2023many}
Z.~Jiang, X.~Dai, S.~Zhao, R.~Wei, and I.~Gray, ``Many-core real-time network-on-chip {I/O} systems for reducing contention and enhancing predictability,'' in \emph{Proceedings of Cyber-Physical Systems and Internet of Things Week 2023}, 2023, pp. 227--233.

\bibitem{hebbache2018shedding}
F.~Hebbache, M.~Jan, F.~Brandner, and L.~Pautet, ``Shedding the shackles of time-division multiplexing,'' in \emph{2018 IEEE Real-Time Systems Symposium (RTSS)}.\hskip 1em plus 0.5em minus 0.4em\relax IEEE, 2018, pp. 456--468.

\bibitem{jiang2022axi}
Z.~Jiang, K.~Yang, N.~Fisher, I.~Gray, N.~C. Audsley, and Z.~Dong, ``Axi-ic$^{RT}$: Towards a real-time axi-interconnect for highly integrated socs,'' \emph{IEEE Transactions on Computers}, vol.~72, no.~3, pp. 786--799, 2022.

\bibitem{herber2014spatial}
C.~Herber, A.~Richter, H.~Rauchfuss, and A.~Herkersdorf, ``Spatial and temporal isolation of virtual can controllers,'' \emph{ACM SIGBED Review}, vol.~11, no.~2, pp. 19--26, 2014.

\end{thebibliography}



\end{document}